\documentclass[11pt]{article}

\usepackage{fullpage}
\usepackage{amsmath}
\usepackage{amsfonts}
\usepackage{amssymb}
\usepackage{amsthm}
\usepackage{mathtools}
\usepackage{booktabs}
\usepackage{graphicx}
\usepackage{natbib}
\usepackage{enumitem}
\usepackage{xurl} 
\usepackage[colorlinks=true,
            linkcolor=red,
            anchorcolor=blue,
            citecolor=blue]{hyperref}

\allowdisplaybreaks

\newtheoremstyle{smileplain}
  {\topsep}{\topsep}{\normalfont}{}{\bfseries}{.}{.5em}{}
\theoremstyle{smileplain}
\newtheorem{theorem}{Theorem}
\newtheorem{lemma}[theorem]{Lemma}
\newtheorem{proposition}[theorem]{Proposition}
\newtheorem{corollary}[theorem]{Corollary}

\newtheorem{remark}[theorem]{Remark}
\numberwithin{theorem}{section}

\numberwithin{equation}{section}
\newcommand{\E}{\mathbb E}
\renewcommand{\P}{\mathbb P}
\newcommand{\R}{\mathbb R}
\newcommand{\N}{\mathbb N}
\newcommand{\cL}{\mathcal L}
\newcommand{\1}{\mathbf 1}
\DeclareMathOperator{\KL}{KL}
\DeclareMathOperator{\kl}{kl}
\DeclareMathOperator{\Ent}{Ent}

\DeclareMathOperator*{\argmax}{arg\,max}
\DeclarePairedDelimiter{\ceil}{\lceil}{\rceil}

\hypersetup{
  pdftitle={Gap Entropy and Almost Instance-Wise Optimal Best-Arm Identification},
  pdfauthor={TODO: author names}
}

\title{\huge Gap Entropy and Almost Instance-Wise\\Optimal Best-Arm Identification}
\author{%
  Jiarui Yao%
  \thanks{Alphabetical order.}
  \thanks{UIUC; e-mail: \texttt{jiarui14@illinois.edu}.}
  \qquad
  Jiaxi Zhao%
  \footnotemark[1]
  \thanks{NUS; e-mail: \texttt{jiaxi.zhao@u.nus.edu}}%
  \qquad
  Xiangxin Zhou%
  \footnotemark[1]
  \thanks{Tencent Hunyuan; e-mail: \texttt{zhouxiangxin1998@gmail.com}}%
}
\date{}

\begin{document}

\maketitle

\begin{abstract}
In the best-arm identification problem, we are given $n$ stochastic arms with unknown means and wish to identify the arm with the largest mean with probability at least $1-\delta$, using as few samples as possible. We consider independent Gaussian rewards with unit variance and means in $[0,1]$. \citet{chen2016open} conjectured that the instance-wise sample complexity of this problem is characterized by the gap entropy, up to an additive term arising from the two-arm problem.

In this paper, we resolve their gap-entropy and almost instance-wise optimality conjectures. For an instance $I$, let $\Delta_{[i]}$ be the gap between the largest and the $i$-th largest mean, let $H(I)=\sum_{i=2}^{n}\Delta_{[i]}^{-2}$, and let $\Ent(I)$ denote the entropy of the normalized complexities of its dyadic gap groups. For every $0<\delta<0.1$, we show that the order-oblivious instance-wise lower bound is
\[
\Theta \Bigl(H(I)\bigl[\log(1/\delta)+\Ent(I)\bigr]\Bigr).
\]
We also give a single $\delta$-correct algorithm with expected sample complexity
\[
O \Bigl(
H(I)\bigl[\log(1/\delta)+\Ent(I)\bigr]
+D\log\bigl(e+\log(e+D)\bigr)
\Bigr),
\qquad D=\Delta_{[2]}^{-2},
\]
without prior knowledge of the gaps. Our lower bound removes the dyadic-gap and monotonicity restrictions of previous work, and our upper bound removes the additional polylogarithmic factor multiplying the two-arm term. Thus, a single algorithm attains the instance-wise lower bound up to an additive two-arm term. The main theorems have been formalized and proved in Lean~4.
\end{abstract}

\medskip

\section{Introduction}
In fixed-confidence best-arm identification, one procedure must identify the arm of largest mean with probability at least $1-\delta$ on every admissible instance, deciding where to sample and when to stop \citep{even2006action,kaufmann2016complexity,garivier2016optimal}. Its difficulty is governed by the separations between means and by the cost of discovering which separations matter \citep{simchowitz2017simulator}.

Let $i_{[1]}$ denote the best-arm label and $\Delta_i$ the gap of arm $i$. The instance complexity is $H=\sum_{i\ne i_{[1]}}\Delta_i^{-2}$, and $H\log(1/\delta)$ captures the confidence requirement \citep{mannor2004sample} but not the cost of discovering an unknown gap scale: even with two arms, an iterated-logarithm obstruction arises \citep{farrell1964}. The order-oblivious instance-wise lower bound of \citet{chen2016open} separates these effects: at each instance it takes the infimum over $\delta$-correct algorithms after averaging over all label permutations. They conjectured that it is characterized by the Shannon entropy of the normalized complexities of the gap groups, and that a single algorithm attains it with only the additive two-arm term.

An algorithm chosen separately for each target instance may have parameters adapted to that instance while remaining correct on every input. In contrast, one algorithm satisfying the bound on all instances must work without these parameters; likewise, a lower bound on a sub-instance does not yield one on the original instance for an arbitrary algorithm. \citet{chen2017towards} obtain the entropy main term with an additional polylogarithmic factor on the two-arm term, and a lower bound on a suitable sub-instance that transfers to the same instance under monotonicity.

We prove both conjectures in the Gaussian model stated by \citet{chen2016open}, for the measurable algorithms of Section~\ref{sec:model} and with the explicit positive truncation of the two-arm term fixed there; Appendix~\ref{bridge:section} shows that the instance-wise lower bound is unchanged when algorithms may use arbitrary standard Borel internal randomness, and that algorithms specified by history-dependent probability kernels satisfy the lower bound. The proof has two main components.

Concurrent work by \citet{aronow2026positive} also resolves both conjectures, with substantial overlap in the lower-bound argument and different upper-bound constructions; Section~\ref{sec:concurrent-work} compares the results and methods.

\paragraph{A lower bound using counts on the original instance.}
After symmetrization, designate a suboptimal arm and raise it to the best mean; a simultaneous change at one other arm makes the resulting distribution with two best arms independent of the designated gap. Both changes are charged, through a latent-variable coupling, to expected numbers of samples on the original instance. Choosing the largest gap among representatives with small normalized expected numbers of samples gives a common comparison distribution. Selecting disjoint stopping events under that distribution bounds the number of such groups and gives a bound on a sum of exponentials; nonnegativity of relative entropy then gives the complexity-weighted entropy term. The argument assumes no monotonicity and uses no bound on the expected number of samples on another instance.

\paragraph{A single algorithm with an additive two-arm term.}
The algorithm performs successive elimination runs with geometrically increasing bounds on the total weight and number of samples, reuses one reference estimate within a scale, and retains at least the empirical upper half after every elimination. Its ordinary error allocation is proportional to the weight assigned to a call; an additional penalty, summable in the run index, applies only to calls with at most seven active arms, so its sampling cost is charged to $D$ rather than to $H$. Correctness is proved by grouping outputs according to fixed nested sets of arms near the original best at the terminal scale: small budgets force many underestimation events, and large budgets are controlled by a bound based on the ordering of empirical means for a call that eliminates all but at most one arm in such a set; the grouping is summable uniformly in the number of arms and gap groups.

\paragraph{Formal verification.}
The definitions of Section~\ref{sec:model}, Theorems~\ref{thm:entropy} and~\ref{thm:universal} with explicit constants, and the results on the algorithm model of Appendix~\ref{bridge:section} have been formalized and verified in Lean~4 with Mathlib\footnote{https://github.com/leanprover-community/mathlib4}; Appendix~\ref{verif:section} states precisely what is verified and how the formal proofs differ from the presentation here. The full Lean 4 code can be found in \url{https://github.com/zhouxiangxin1998/GapEntropy}.

\section{Problem formulation and main results}\label{sec:model}
\subsection{Observation model and statistical objective}
An instance is a vector $I=(\mu_1,\ldots,\mu_n)\in[0,1]^n$, where $n\ge2$ and the maximum is unique. Write $[n]=\{1,\ldots,n\}$. Arm $i\in[n]$ produces independent $\mathcal N(\mu_i,1)$ observations. Let $i_{[r]}$ be the label of the arm with the $r$th largest mean, breaking ties between suboptimal arms by label, and write $\mu_{[r]}=\mu_{i_{[r]}}$. Thus
\[
 i_{[1]}=\argmax_{i\in[n]}\mu_i,\qquad
 \Delta_i=\mu_{[1]}-\mu_i,\qquad
 \Delta_{[r]}=\Delta_{i_{[r]}}=\mu_{[1]}-\mu_{[r]}.
\]
The ordered gaps satisfy $0=\Delta_{[1]}<\Delta_{[2]}\le\cdots\le\Delta_{[n]}\le1$. An \emph{algorithm} has an internal random seed, a sequence of independent standard Gaussian variables independent of all rewards. At each time it applies a measurable rule to its finite observed history (the labels sampled so far and their observations) and its seed, and either requests one observation from an arm or returns a label; a return is absorbing. Let $\tau$ be the first return time, with $\tau=\infty$ if no label is returned, and $\widehat{\imath}$ the returned label. Sample complexity means the unconditional expectation $T_A(I)=\E_I\tau$, with no computational-time requirement. An algorithm is $\delta$-correct if $\P_I\bigl(\tau<\infty,\widehat{\imath}=i_{[1]}\bigr)\ge1-\delta$ for every admissible instance at every arm count $n\ge2$. An algorithm on singleton inputs simply returns the only arm. Appendix~\ref{bridge:section} shows that the instance-wise lower bound below is unchanged if the internal random seed is drawn from an arbitrary standard Borel probability space, and that algorithms specified by history-dependent probability kernels satisfy the lower bound; the upper-bound algorithms use no internal randomness.

Time counts observations: until return, the algorithm requests one observation at each integer time, so a nonreturning path has an infinite observation stream. Let $A_t$ be the label sampled at time $t$. Conditional on the information available before that sample, the observation has law $\mathcal N(\mu_{A_t},1)$. Define $N_i(t)=\sum_{s=1}^{t\wedge\tau}\1\{A_s=i\}$ and let $N_i(\tau)$ denote the total count, allowing $+\infty$. Then $\sum_iN_i(t)=t\wedge\tau$ and $T_A(I)=\sum_i\E_I N_i(\tau)$ by monotone convergence. Thus the complexity charges every observation, including observations on erroneous trajectories. Correctness alone does not imply finite expectation; the upper bounds establish both properties separately.

\subsection{The instance-wise lower bound and the gap groups}

For a permutation $\pi$ of the labels, let $\pi I$ denote the instance obtained by moving arm $i$ to label $\pi(i)$. The order-oblivious instance-wise lower bound of \citet[Definition~3.1]{chen2016open} is
\begin{equation}\label{eq:benchmark}
 \cL(I,\delta)=\inf_{A:\ A\text{ is }\delta\text{-correct}}
 \frac1{n!}\sum_{\pi\in\mathfrak S_n}T_A(\pi I).
\end{equation}
The infimum is taken separately at each $I$. Algorithms in this infimum may have different fixed parameters chosen for different target instances, but each must remain correct on every admissible input. The single algorithm below takes only the arms and $\delta$ as inputs.

For integers $k\ge0$, define
\begin{equation}\label{eq:gap-entropy}
 G_k=\bigl\{i\ne i_{[1]}:2^{-k}\le\Delta_i<2^{-k+1}\bigr\},\quad
 H_k=\sum_{i\in G_k}\Delta_i^{-2},\quad
 H=\sum_kH_k,\quad p_k=H_k/H.
\end{equation}
The gap $1$ belongs to $G_0$. Zero-mass groups are omitted from entropy sums, and $\Ent(I)=\sum_{k:p_k>0}p_k\log(1/p_k)$. For any discrete distribution $q$, $\Ent(q)$ denotes its Shannon entropy, with $0\log(1/0)=0$; for discrete random variables, $\Ent$ denotes the entropy of their law. Throughout,
\begin{equation}\label{eq:notation-main}
 D=\Delta_{[2]}^{-2},\qquad L=\log(1/\delta),\qquad
 \ell(D)=\log\bigl(e+\log(e+D)\bigr).
\end{equation}
When the instance is fixed, we suppress its argument in $H=H(I)$, $D=D(I)$, and the gap-group quantities. All unmarked logarithms are natural. Constants in $O$, $\Omega$, and $\Theta$ are universal, independent of $n$, $\delta$, and the instance. We write $\KL$ for relative entropy and $\kl$ for its Bernoulli specialization.

Table~\ref{tab:notation} summarizes the notation, with $n_k=|G_k|$ and $K=\max\{k:H_k>0\}$.
\begin{table}[t]
\caption{Principal instance quantities and their roles.}
\label{tab:notation}
\begin{center}
\begin{tabular}{@{}p{0.23\linewidth}p{0.73\linewidth}@{}}
\toprule
Notation & Meaning \\
\midrule
$i_{[r]}$, $\mu_{[r]}$ & Label and mean of the arm with the $r$th largest mean; $i_{[1]}$ is the best-arm label. \\
$\Delta_i$, $\Delta_{[r]}$, $D$ & Gap of label $i$, gap at rank $r$, and inverse squared smallest positive gap. \\
$H_k$, $H$ & Total complexity of group $k$ and complexity of the instance. \\
$p_k$, $\Ent(I)$ & Normalized group complexities and their Shannon entropy. \\
$n_k$, $K$ & Number of suboptimal arms in group $k$ and the largest index of a nonempty group. \\
$L$, $\ell(D)$ & Logarithmic confidence term and truncated iterated logarithm. \\
\bottomrule
\end{tabular}
\end{center}
\end{table}

Gap entropy weights each group by its total complexity rather than its number of arms. If $b$ groups are nonempty, then $0\le\Ent(I)\le\log b$, with the upper equality when all $H_k$ are equal. A single nonempty group has zero entropy. For any integer $K\ge1$, take best mean $1$ and $n_k=4^{K-k}$ arms of gap $2^{-k}$ for each $1\le k\le K$. Then $H_k=4^K$, $H=K4^K$, and $\Ent(I)=\log K$, despite the very different group sizes. This example isolates the cost of exploring several groups with comparable total complexities. These quantities describe the instance and are not supplied to the single algorithm.

\subsection{Main results}

\begin{theorem}[Resolution of the gap-entropy conjecture]\label{thm:entropy}
For every instance in the model above and every $0<\delta<0.1$,
\begin{equation}\label{eq:main-entropy}
 \cL(I,\delta)=\Theta\!\Bigl(H(I)\bigl[\log(1/\delta)+\Ent(I)\bigr]\Bigr).
\end{equation}
\end{theorem}

\begin{theorem}[A single almost instance-wise optimal algorithm]\label{thm:universal}
For every $0<\delta<0.1$ there is one $\delta$-correct algorithm whose expected stopping time is finite on every admissible instance and which satisfies
\begin{equation}\label{eq:main-universal}
 T_A(I)=O\!\Bigl(H(I)\bigl[\log(1/\delta)+\Ent(I)\bigr]+D\ell(D)\Bigr).
\end{equation}
It receives no estimate of any gap, $H$, or $\Ent(I)$.
\end{theorem}

Together these theorems imply $T_A(I)=O\bigl(\cL(I,\delta)+D\ell(D)\bigr)$, establishing Conjectures~3.5 and~3.2 of \citet{chen2016open} for the algorithms defined above, under the following convention. The additive term of Conjecture~3.2 is written there as $\Delta_{[2]}^{-2}\ln\ln\Delta_{[2]}^{-1}$, which is negative for $\Delta_{[2]}\in(e^{-1},1)$ and undefined at the admissible gap $\Delta_{[2]}=1$; a positive convention is therefore required. We use $g(\Delta)=\max\bigl\{1,\log\log(1/\Delta)\bigr\}$ for $0<\Delta<1$ and $g(1)=1$. For every $0<\Delta\le1$, $g(\Delta)\le\ell(\Delta^{-2})\le7g(\Delta)$ (see \eqref{bridge:sandwich}), so the two positive truncations give equivalent statements; as $\Delta_{[2]}\to0$ both are $\log\log(1/\Delta_{[2]})\bigl(1+o(1)\bigr)$. A uniform $O(D)$ comparison with the untruncated expression does not hold near $\Delta_{[2]}=1$. Corollary~\ref{bridge:regimes} makes the two regimes explicit: for $\Delta_{[2]}\le e^{-e}$ the bound holds with the untruncated $\Delta_{[2]}^{-2}\log\log(1/\Delta_{[2]})$, and for $\Delta_{[2]}>e^{-e}$ the additive term is absorbed into $O\bigl(\cL(I,\delta)\bigr)$.

The proofs give explicit constants. Theorem~\ref{lower:theorem} shows
\[
 \cL(I,\delta)\ge H(I)\max\bigl\{2\kl(1-\delta,\delta),\Ent(I)/4\bigr\}
 \ge\tfrac15H(I)\bigl[\log(1/\delta)+\Ent(I)\bigr],
\]
where $\kl$ is the binary relative entropy; the universal constants in the upper bounds are large and are not optimized.

Section~\ref{overview:entropy} presents both directions of Theorem~\ref{thm:entropy}; Section~\ref{overviewu:section} gives the single algorithm and the key arguments for Theorem~\ref{thm:universal}. Appendix~\ref{sec:prelim} collects the probabilistic tools and elimination guarantees used in these proofs. Complete proofs appear in Appendices~\ref{lower:section}--\ref{cost:section}. Appendix~\ref{app:section} records consequences and the normalization to general known variances, and Appendix~\ref{verif:section} describes the machine-checked verification of the main statements. The instance-wise upper bound is independent of the construction of the single algorithm.

\section{Related work}
\label{sec:related-work}

\paragraph{Gap-dependent complexity and elimination.}
Best-arm identification separates the cost of distinguishing a fixed mean
vector from the additional cost of adapting to unknown gaps.
\citet{mannor2004sample} established the classical gap-dependent
lower-bound scale $H(I)\log(1/\delta)$.
The sequential-testing lower bound of \citet{farrell1964} supplies the
iterated-logarithmic obstruction in the two-arm problem as the gap tends
to zero. This is an asymptotic obstruction across gap values, rather than
an additional lower bound that must hold at every fixed instance for
every $\delta$-correct algorithm. That distinction motivates the
order-oblivious instance-wise lower bound and the additive allowance in
\citet{chen2016open}.

Elimination methods provide a basic route to efficient exploration.
Median elimination and action elimination
\citep{even2002pac,even2006action} use confidence guarantees to
progressively discard unpromising arms. Median elimination finds an
$\varepsilon$-optimal arm using
$O\bigl(n\varepsilon^{-2}\log(1/\delta)\bigr)$ samples; our reference-selection
routine uses this principle, with its Gaussian guarantee proved in the
paper. \citet{karnin2013almost} and \citet{jamieson2014lil} obtained
bounds with iterated-logarithmic dependence on the individual gaps.
The unified gap-based exploration framework of
\citet{gabillon2012unified} gives related sampling strategies for
fixed-confidence and fixed-budget objectives; KL-based confidence
intervals also sharpen subset-selection guarantees
\citep{kaufmann2013information}.
These bounds explain why eliminating the cost of scale adaptation from
every suboptimal arm is a substantive part of the almost instance-wise
optimality question.

\paragraph{Gap entropy and instance-wise optimality.}
\citet{chen2015optimal} improved earlier upper bounds and used a
refined sign-testing lower bound to obtain best-arm lower bounds beyond
the classical confidence term.
\citet{chen2016open} defined gap entropy using the normalized complexity
of gap groups and posed two precise conjectures. Their
Conjecture~3.5 asks whether the order-oblivious instance-wise lower bound
is $\Theta\Bigl(H(I)\bigl[\log(1/\delta)+\Ent(I)\bigr]\Bigr)$.
Their Conjecture~3.2 asks for one $\delta$-correct algorithm whose
expected number of samples is within a constant factor of that instance-wise lower bound
plus the two-arm iterated-logarithmic term. The infimum defining the
instance-wise lower bound allows a different $\delta$-correct algorithm for each
target instance, whereas the latter conjecture requires a single
algorithm that adapts to all instances.

The direct follow-up of \citet{chen2017towards} proves both an upper
bound and a restricted form of the entropy lower bound.
Their Theorem~1.11 gives expected sample complexity
\begin{equation*}
 O\!\Bigl(
 H(I)\bigl[\log(1/\delta)+\Ent(I)\bigr]
 +\Delta_{[2]}^{-2}\log\log\Delta_{[2]}^{-1}
       \operatorname{polylog}(n,\delta^{-1})
 \Bigr),
\end{equation*}
with the customary nonnegative interpretation of the iterated
logarithm. Their Theorem~1.12 states that, for every instance $I$ whose
suboptimal gaps are all of the form $2^{-k}$ with $k\ge1$ an integer
(their \emph{discrete} instances), every $\delta$-correct algorithm
with $\delta<0.01$ incurs
$\Omega\Bigl(H(I)\bigl[\log(1/\delta)+\Ent(I)\bigr]\Bigr)$ samples on
\emph{some sub-instance} $I'$ obtained by deleting suboptimal arms.
Their Corollary~1.13 transfers the bound to $I$ itself for monotone
algorithms, meaning algorithms whose expected number of samples cannot
increase when suboptimal arms are removed. The sub-instance statement
alone does not give the lower bound on the original instance for an
arbitrary algorithm.

For the algorithms of Section~\ref{sec:model}, the results proved here
address these specific differences from
\citet{chen2017towards}: the lower bound concerns the original instance
under permutation averaging, permits arbitrary positive gaps and
$0<\delta<0.1$, and imposes no monotonicity assumption. The uniform
upper bound has the entropy main term and an additive two-arm term
without the displayed $\operatorname{polylog}(n,\delta^{-1})$ factor,
under the positive truncation convention of Section~\ref{sec:model}.

\paragraph{Information bounds and the confidence limit.}
\citet{kaufmann2014testing} analyze distribution-dependent complexities
in two-arm testing. \citet{kaufmann2016complexity} develop a general
change-of-measure inequality for adaptive bandit experiments and
lower bounds for identifying the best arms.
\citet{garivier2016optimal} characterize the optimal allocation through
an information optimization problem and prove asymptotic optimality of
Track-and-Stop as $\delta\to0$ for a fixed instance.
This limit determines the leading confidence cost, but does not by
itself control the additional dependence on unknown gaps uniformly
over instances. \citet{simchowitz2017simulator} make the distinction
between moderate confidence and the vanishing-error limit explicit:
their simulator method yields adaptive-sampling lower bounds that
capture costs of learning where to sample. Our lower bound likewise
addresses adaptation, using the original-instance counts and the
common comparison distribution described in Section~\ref{overview:entropy}.
\citet{degenne2019games} obtain finite-confidence guarantees for more
general pure-exploration problems by viewing the information
optimization as a game and using no-regret learners.

\paragraph{Top-two sampling.}
An alternative to elimination is to sample from a pair of candidate
arms. \citet{russo2016simple} introduces Bayesian top-two allocation
rules and analyzes their posterior concentration rates.
\citet{shang2020fixed} establish fixed-confidence sample-complexity
guarantees for Top-Two Thompson Sampling and a transportation-cost
variant with Gaussian rewards.
\citet{jourdan2022top} extend the analysis to bounded reward
distributions and permit alternative choices of the leading arm.
\citet{jourdan2023nonasymptotic} give nonasymptotic bounds on expected
sample complexity for a UCB-based top-two algorithm, while
\citet{you2023information} develop information-directed selection
between the two candidates and obtain asymptotic optimality for
Gaussian best-arm identification. These results study how to learn
effective sampling allocations; the particular target here is the
gap-entropy benchmark with only an additive two-arm adaptation term.

\paragraph{Sequential stopping and time-uniform inference.}
Confidence sequences provide another perspective on the
iterated-logarithmic cost of sequential inference.
\citet{howard2021sequences} construct nonasymptotic confidence
sequences valid over an unbounded time horizon, including boundaries
with iterated-logarithmic behavior.
\citet{kaufmann2021mixture} derive time-uniform deviation bounds
under adaptive bandit sampling and apply them to likelihood-ratio
stopping rules. Our upper-bound proof instead allocates failure
probabilities across finite elimination calls; the challenge is to
make that allocation adapt to the unknown complexity of the gap groups.

\paragraph{Fixed budgets and structured exploration.}
With a fixed sampling budget, the objective is to control the error
of the final recommendation. Successive Rejects
\citep{audibert2010best} and Sequential Halving
\citep{karnin2013almost} are central elimination methods in this
setting. The lower bounds of \citet{carpentier2016tight} show that
adaptation to an unknown instance can require a logarithmic factor in
the number of arms in fixed-budget guarantees. These results concern
a different performance criterion from the expected stopping time
studied here.
Other extensions change the structure of the identification task:
\citet{chen2014combinatorial} study the selection of an optimal
feasible subset, and linear-bandit methods exploit information shared
across arm means \citep{soare2014linear,tao2018linear,jedra2020optimal}.
\citet{huang2017structured} consider actions whose values depend on
noisy component observations, motivated by minimax game search;
\citet{degenne2019multiple} treat pure-exploration tasks with multiple
correct answers. Our theorems concern independent Gaussian arms with
a unique best arm; extensions of gap entropy to these structured
objectives would require additional analysis.

\subsection{Concurrent resolution of the two conjectures}
\label{sec:concurrent-work}

\citet{aronow2026positive}, posted on September~9, 2026, also resolve
the gap-entropy and almost instance-wise optimality conjectures.
We briefly compare with their results here. Their Theorems~1.1--1.2 give
the same orders as Theorems~\ref{thm:entropy}--\ref{thm:universal}
on Gaussian instances; their upper bounds additionally cover independent
$1$-sub-Gaussian rewards. Our theorems are stated for Gaussian rewards.
The opposite endpoint conventions for dyadic groups change entropy by
at most $\log2$: each cell of either partition intersects at most two
cells of the other, so both conditional entropies are at most $\log2$.
Their regularized two-arm term is equivalent to $D\ell(D)$ up to
absolute constants. These conventions therefore leave the rates unchanged.

\paragraph{Lower bound.}
The central mechanism is shared: permutation symmetrization, a change
of at most two arm means producing a common tied reference law, and
disjoint sample-count events followed by an entropy inequality.
Both arguments charge the changes to counts on the original instance.
Their Section~2 uses averages within gap groups, normalized by the
dyadic sampling scale. Our Appendix~\ref{lower:section} instead works
with exact-gap counts $\alpha_d=d^2t_d$ and selects a minimizing
representative in each group. Our packing step also uses the
$\delta$ bound on the total reference probability of the selected
return events, yielding $\sum_k e^{-4\alpha_k}<10\delta$.
These are differences within a closely related proof strategy.

\paragraph{Instance-wise upper bound.}
Their Section~3 uses mixtures of exponential supermartingales around
a target-dependent threshold between the two largest means, followed
by a confirmation test. Our Appendix~\ref{pointwise:section} uses
only the target's group sizes to specify a finite elimination schedule
and smoothed confidence weights. Its thresholds are estimated from
the active arms, and prescribed size reductions bound every run's
cost on every input. Both constructions repeat trials and interleave
a globally correct fallback to ensure termination away from the target.

\paragraph{A single algorithm.}
Both methods use increasing budgets, confidence allocation proportional
to sampling weights, and bounds on joint premature eliminations.
Their Sections~4--6 retain the Median--Fraction--Eliminate structure
of \citet{chen2017towards}; the fraction thresholds depend on the
number of rounds remaining. Their error analysis separates hard and
easy arms according to the stage budget. Our algorithm
(Appendix~\ref{alg:section}) reuses one reference estimate per scale,
retains at least the empirical upper half, and adds a summable run-index
penalty only when at most seven arms remain. Our correctness proof
groups outputs by fixed nested sets of near-optimal arms and combines
stochastic domination with an empirical-ordering bound
(Appendix~\ref{sound:section}). These features provide a different
way to obtain the additive two-arm cost.

Our contribution also includes an explicit measurable-policy treatment
and Lean~4 verification of the main statements and policy-model
extensions. Appendix~\ref{verif:section} specifies the scope of that
verification, including the differences between the formal witnesses
and the algorithms presented here.

\section{Characterizing the order-oblivious instance-wise lower bound}
\label{overview:entropy}

We explain the two directions of Theorem~\ref{thm:entropy}. Throughout,
$I$ is a Gaussian instance with a unique best arm in the model of
Section~\ref{sec:model}, and $0<\delta<1/10$. Complete proofs appear in
Appendices~\ref{lower:section} and~\ref{pointwise:section}.

\paragraph{A common comparison distribution using samples from the original instance.}
Randomly relabeling any $\delta$-correct algorithm makes its expected
cost on $I$ equal to its permutation average. For this symmetrized
algorithm, let $t_d$ be the expected number of samples from an individual arm with
exact gap $d$, and set $\alpha_d=d^2t_d$. Permutation equivariance
gives every arm with gap $d$ the same expected count on every labeling
of the target; we may assume these counts are finite, since otherwise
there is nothing to prove.

Let $P_d$ place a gap-$d$ arm at label $1$ and uniformly label the
remaining target multiset. Let $Q_a$ instead give label $1$ mean
$\mu_{[1]}$, with the remaining labels containing the target multiset
minus one gap-$a$ arm. Thus $Q_a$ has two best arms. Correctness
implies $Q_a(\text{return a label other than }1)\le\delta$: lower its
other best mean by $\varepsilon$, apply correctness with label $1$
uniquely best, and pass first to the finite-transcript limit
$\varepsilon\downarrow0$, then to infinite time. Lowering that mean
keeps the perturbation admissible even when $\mu_{[1]}=1$; no
termination assumption at the tie is needed.

For $a\ne d$, couple $P_d$ to $Q_a$ by raising label $1$ to
$\mu_{[1]}$ and changing a uniformly selected gap-$a$ arm to gap $d$.
The resulting background is uniformly labeled, including when gaps
have multiplicities. The second change restores the gap multiplicity
lost at label $1$. Uniform labeling makes the comparison law depend
only on $a$, so different source scales can share it. Keep the source
assignment and selected label
latent, with the same prior in both experiments. Stop each transcript
at return or after $m$ total observations, denoting its law by a
superscript $(m)$. The finite-horizon Gaussian KL bound
(Lemma~\ref{lem:gaussian-kl}), followed by discarding the latent variables, gives
\begin{equation}
\label{overview:coupling}
 \KL\bigl(P_d^{(m)}\bigm\Vert Q_a^{(m)}\bigr)
 \le\tfrac12\bigl[d^2t_d+(a-d)^2t_a\bigr]
 \le\tfrac12(\alpha_d+\alpha_a),\qquad a\ge d.
\end{equation}
When $a=d$, only label $1$ changes and the bound is $\alpha_d/2$.
Binary data processing then gives $\alpha_d\ge
2\kl(1-\delta,\delta)\ge L$. The hidden selection is independent of
observations conditional on the source assignment. Consequently all
costs in~\eqref{overview:coupling} are measured on the original target,
without a monotonicity assumption or a transfer of expected sample complexity between
different instances.

\paragraph{Disjoint stopping events imply the entropy lower bound.}
We adapt the stopping-event argument of
\citet[Appendices~D.2--D.3]{chen2017towards} to the comparison above.
In each nonempty group choose a gap minimizing $\alpha_d$, and denote
the minimum by $\alpha_k$. Among representatives with $\alpha_k\le x$,
choose the largest gap $a$, whenever this set is nonempty.
Lemma~\ref{lower:windows} gives corresponding stopping events $E_d$
with $Q_a(E_d)\ge e^{-2x}/4$ under this common reference distribution.
Selecting disjoint sample-count windows and using
$Q_a(\text{return a label other than }1)\le\delta$ bounds the number
of these representatives. Integrating that bound yields
\begin{equation}
\label{overview:packing}
 Z:=\sum_k e^{-4\alpha_k}<10\delta<1.
\end{equation}
The reference gap depends on $x$; Appendix~\ref{lower:section} gives
the window selection and integration.
Nonnegativity of relative entropy gives
$4\sum_kp_k\alpha_k\ge\Ent(I)-\log Z\ge\Ent(I)$.
The target cost is at least $H\sum_kp_k\alpha_k$ by groupwise
minimization. Combining this with
$\alpha_k\ge2\kl(1-\delta,\delta)$ gives the lower bound
$H\max\bigl\{2\kl(1-\delta,\delta),\Ent(I)/4\bigr\}\ge H\bigl[L+\Ent(I)\bigr]/5$.

\paragraph{An algorithm with a fixed target instance that remains $\delta$-correct.}
For the upper bound, fix the group sizes of the target instance as algorithmic constants. Let $n$ be the number of arms in that instance and $K$ the largest index of a nonempty group.
This is permitted inside the pointwise infimum defining $\cL$, but
correctness must still hold on every input. Use weights defined from the group sizes
$W_k=4^k\bigl(1+\sum_{j\ge k}|G_j|\bigr)$, for $0\le k\le K$,
and $q_k=W_k/\sum_{j=0}^K W_j$.
Lemma~\ref{lem:smoothing} gives total weight $O(H)$ and normalized
entropy $O\bigl(1+\Ent(I)\bigr)$. A finite run allocates confidence
proportional to $q_k$, with geometric subdivision among same-scale
raw threshold-elimination calls, returning $R_0$ without enlargement. It aborts whenever a required halving
fails. The group sizes determine how many arms may remain at each
scale; the final call at the target's finest scale returns only if its
retained set is a singleton. At tolerance $d$, a call selects a
reference arm by median elimination \citep{even2002pac}, estimates
its mean by $z$, and retains arms whose fresh empirical means satisfy
$x_i\ge z-d/2$. For $s$ active arms and error budget $\alpha$,
Lemmas~\ref{lem:pac}--\ref{lem:elimination} give a deterministic
$O\bigl(sd^{-2}\log(1/\alpha)\bigr)$ sample budget and the required best-arm
retention and halving guarantees. The imposed size reductions
give a cost bound even after statistical failures. On every input with
$n$ arms, each run trajectory has deterministic cost
\begin{equation}
\label{overview:pointwise-cost}
 B_I(\varepsilon)
 =O\!\Biggl(\sum_{k=0}^K W_k\log\frac1{\varepsilon q_k}\Biggr)
 =O\Bigl(H\bigl[\log(1/\varepsilon)+\Ent(I)\bigr]\Bigr).
\end{equation}
The run's incorrect-return probability is at most $\varepsilon$
on every such input, while it succeeds with probability at least
$1-\varepsilon$ on every labeling of the target.

Repeat fresh runs with $\varepsilon_t=\delta2^{-(t+2)}$ and
interleave an independent $\delta/2$-correct confidence-interval
algorithm having finite expected stopping time on every input with a unique best arm. Other arm counts use the fallback alone. This fallback supplies termination even when the input differs from the target instance. Errors sum to at most $\delta$, and interleaving
costs at most a factor two. On the target, run $t$ is reached with
probability at most $2^{-t}$, so summing
\eqref{overview:pointwise-cost} gives $O\Bigl(H\bigl[L+\Ent(I)\bigr]\Bigr)$ samples in expectation, uniformly over its labelings.

\section{A single algorithm}
\label{overviewu:section}

Theorem~\ref{thm:universal} requires one algorithm that works without
knowing the instance's gaps, complexity, or entropy. We describe the
construction and the main proof arguments; the complete procedure,
correctness proof, and analysis of the expected sample complexity are
Appendices~\ref{alg:section}, \ref{sound:section}, and
\ref{cost:section}. The assumptions remain independent
unit-variance Gaussian observations, means in $[0,1]$, a unique best
arm, and $0<\delta<0.1$.

\subsection{Sample budgets and confidence allocation}
\label{overviewu:algorithm}

Write $L=\log(1/\delta)$ and $\gamma=\delta/1024$. Run
$j=0,1,\ldots$ starts with all arms active and fresh observations, and
has deterministic budgets
\begin{equation}
 h_j=2^j,\qquad M_j=1024h_j,\qquad
 Q_j=\ceil[\big]{C_Qh_jL}.
 \label{overviewu:caps}
\end{equation}
The bound $M_j$ controls the sum of the weights of calls, whereas $Q_j$ bounds the number of samples;
$C_Q$ is the absolute constant $C_*$ of Appendix~\ref{cost:section},
determined by the constants of Lemmas~\ref{lem:pac}
and~\ref{lem:elimination}; any larger absolute constant also works. Within a run, tolerances are $d=2^{-k}$, with
$k=0$ initially; the scale index increases only when elimination stalls.
For a call on the current active set $S$, of size $s\ge2$, define the weight $w=sd^{-2}$, the factor in its sample bound before the confidence logarithm. Its confidence parameter is
\begin{equation}
 \alpha=\gamma\frac{w}{M_j}
 \begin{cases}
 (j+1)^{-2},&2\le s\le7,\\
 1,&s\ge8.
 \end{cases}
 \label{overviewu:confidence}
\end{equation}
The size in this rule is the current active size at every call, so the
branch can change after any call, including after a stall.
Thus only calls with at most seven active arms incur the additional sample cost of this confidence allocation.

We use standard concentration and median-elimination guarantees
\citep[Facts~5.1--5.2]{chen2017towards}; Appendix~\ref{sec:prelim}
specifies the deterministic sample counts and the modified elimination rule.
At scale entry, median elimination selects a reference arm within
$d/8$ of the active maximum, with failure probability at most
$\alpha/16$. Independently estimate its mean by $z$, with accuracy
$d/16$ and failure budget
$\beta_{j,k}=\delta/\bigl[64(j+1)^2(k+1)^2\bigr]$. At every call obtain
fresh independent estimates $x_i$ for all active arms. The sample
counts are
\begin{equation}
 m_z=\ceil[\big]{512d^{-2}\log(2/\beta_{j,k})},\qquad
 m=\ceil[\big]{512d^{-2}\log(128/\alpha)}.
 \label{overviewu:samples}
\end{equation}
The numerical reference estimate is reused throughout the scale, even
if its arm is removed, so it is paid for once per scale.

The threshold retained set is $R_0=\{i:x_i\ge z-d/2\}$. Enlarge it by adding
arms in decreasing empirical-mean order until at least $\ceil[\big]{s/2}$
arms remain, so the result contains both $R_0$ and the empirical
upper half. A singleton is returned only after this enlargement. If the enlarged
set has size $\ceil[\big]{s/2}$, continue at the same scale; otherwise
advance to the next scale.

Before any sampling for a call,
check its weight and its entire deterministic sample count against
the remaining budgets; a failed reservation aborts the run and starts
the next one. Every trajectory therefore satisfies
$\sum_t w_t\le M_j$, $\sum_t\alpha_t\le\gamma$, and
$T_j\le Q_j$, including trajectories that eliminate the best arm.

\subsection{Correctness via sets of near-optimal arms}
\label{overviewu:soundness}

Summing a constant error allowance over infinitely many runs would
not establish correctness. Our analysis instead separates errors in calls with at most seven arms from errors that require a large batch of eliminations. The
key implication uses the original best mean $\mu_{[1]}$, even after
its arm has disappeared: the reference upper bound implies $z-d/2\le\mu_{[1]}-7d/16$, so removing an arm with
$\Delta_i\le d/8$ requires an \emph{underestimation event},
$x_i-\mu_i<-5d/16$, whose conditional probability is at most
$(\alpha/128)^{25}$ by Gaussian tails. All reference-estimation
failures together have probability below $\delta/16$.

Set $\eta=\gamma/128$ and $p=2\eta^{25}$.
Lemma~\ref{sound:domination} couples all underestimation events in one adaptive
run to independent Bernoulli indicators, one per original arm,
with parameters at most $p$. It uses independent Poisson processes
and a common clock whose predictable increments realize the exact
Gaussian underestimation probabilities. The bound on the total weight limits its total elapsed
clock time by $p$. This construction supplies independence on fixed
disjoint arm sets despite adaptive sample counts. It never conditions
on future reference upper bounds or on which set the algorithm ultimately
selects.

For calls with at most seven arms, the penalty in
\eqref{overviewu:confidence} gives a direct global bound: their
total probability of underestimation in run $j$ is at most
$7\eta^{25}(j+1)^{-50}$, because
$\sum_t\alpha_t^{25}\le\bigl(\sum_t\alpha_t\bigr)^{25}$, and these errors are
charged once over the entire stream, for total probability at most
$14\eta^{25}$.

Outside reference failures and the charged errors in calls with at most seven arms, associate
an output at its final tolerance
$d_*$ with the set of near-optimal arms at that scale
\begin{equation}
 B=\{i:\Delta_i<d_*/8\},\qquad
 r=|B|,\quad h=\sum_{i\notin B}\Delta_i^{-2},
 \qquad A_r=rp^{\max(1,r-1)}.
 \label{overviewu:core}
\end{equation}
Although the final scale is random, the possible sets $B$ form a fixed
nested family determined by the instance, with at most one set of each cardinality. An incorrect singleton
requires underestimation events on all but at most one arm in $B$, or on the original
best when $r=1$, giving probability at most $A_r$ in one run.

Fix $B$ before applying the coupling and group runs according to the bound $M$ on their total weight. For $M\le h/192$, eliminations of arms outside $B$ at tolerances $d_t\le8\Delta_i$, together with a possible output outside $B$, account for total complexity at most $64M+32M$. Arms outside $B$ eliminated at larger tolerances therefore have total complexity at least $h/2$. The independent Bernoulli variables for arms in $B$ and outside $B$ give the tail bound
$A_r(2ep)^{h/(64M)}$, which sums geometrically over these budgets.
Constantly many budgets lie in $h/192<M<h$.

For $M\ge h>0$, an incorrect output outside the already charged
events requires a call with at least eight arms that reduces the
active arms in $B$ to at most one, or removes its sole remaining arm.
Lemma~\ref{sound:collapse} handles this entire batch, which may begin with many arms in $B$ and therefore requires a bound for their joint elimination: retaining the empirical upper half forces either underestimation of many arms in $B$ or overestimation of many arms outside $B$ with sufficiently large gaps, with probability at most
$3\eta^4(h/M)^3$. Taking the minimum of this bound and $A_r$ gives
$\sqrt{3A_r}\eta^2(h/M)^{3/2}$, again summable over budgets, without
any independence between the underestimation events for arms in $B$ and the elimination event.
The lemma sums the conditional probabilities of elimination using the bound on the total weight of calls, so stalled scales add no union-bound factor.
When $h=0$, the retention rule makes such a large-set elimination event impossible.
Finally, $\sum_rA_r=O(p)$ and
$\sum_r\sqrt{A_r}=O\bigl(\sqrt p\bigr)$ make the union bound over this fixed family finite
without a factor for empty scales. Together with the two global
charges, the explicit constants give error probability below $\delta$.

\subsection{An unconditional expected sample complexity bound}
\label{overviewu:cost}

Correctness alone does not control stopping time. In a fresh run,
accurate references and the usual best-arm retention and elimination
events hold with probability at least $1-q_0$, where
$q_0=\gamma+\delta/16<1/2$. On this good event, a scale stalls
only when the active size is within a constant factor of the number
of arms closer than its tolerance. If $n_t$ counts gaps in group
$t$ and $K$ is the largest index of a nonempty group, the weights
$W_k=4^k\bigl(1+\sum_{t\ge k}n_t\bigr)$, for $0\le k\le K$, satisfy
$H\le\sum_{k=0}^K W_k<32H/3$. Their normalized entropy is
$O\bigl(\Ent(I)+1\bigr)$ by Lemma~\ref{lem:smoothing}. Repeated halvings
give the upper bound on the weights of calls $w_{k,r}\le4W_k(3/4)^r$, so runs on this good event with $h_j\ge H$ fit the weight bound and finish by scale $K$.

Applying the logarithmic factor in the sample bound to these weights yields baseline
cost $O\Bigl(H\bigl[L+\Ent(I)+1+\log(h_j/H)\bigr]\Bigr)$. Calls with at most seven
arms have total input size at most $7+4+2=13$ per scale, so their
extra confidence cost is $O\bigl(4^k\log(j+1)\bigr)$ at scale $k$ and
$O\bigl(D\log(j+1)\bigr)$ in total.
One reference estimate per scale adds
$O\Bigl(D\bigl[L+\log(j+1)+\ell(D)\bigr]\Bigr)$.

Put $\Psi=H\bigl[L+\Ent(I)+1\bigr]+D\ell(D)$ and
$j_* =\ceil[\big]{\log_2(\Psi/L)}$. For $C_Q\ge C_*$, every run $j\ge j_*$ on the good event fits its entire sample budget, including each prefix
and its next complete reservation. The choice is noncircular: $C_*$
bounds the declared cost of runs constrained only by the weight bound, computed without the
sample budget. Independence of the fresh observations bounds the probability of reaching run
$j_*+t$ by $q_0^t$. Deterministic budgets control all other trajectories, hence
\begin{equation}
 \E_I\tau\le\sum_{j<j_*}Q_j+
                 \sum_{t\ge0}Q_{j_*+t}q_0^t
 =O(\Psi).
 \label{overviewu:expectation}
\end{equation}
The latter series converges because $2q_0<1$. This proves finite
expected stopping time, almost-sure termination, and the bound in
Theorem~\ref{thm:universal}, with no instance parameters supplied to
the algorithm.

\section{Conclusion}
For independent unit-variance Gaussian arms with means in $[0,1]$ and the measurable algorithms of Section~\ref{sec:model}, the order-oblivious instance-wise lower bound is $\Theta\Bigl(H\bigl[\log(1/\delta)+\Ent(I)\bigr]\Bigr)$. A single algorithm matches this bound up to an additive \(O\Bigl(D\log\bigl(e+\log(e+D)\bigr)\Bigr)\) term. The lower bound follows from a change of distribution to instances with two best arms whose KL divergence is bounded using expected numbers of samples on the original instance; the upper bound for a single algorithm from elimination that retains the arms with the largest empirical means and a summation over fixed sets of near-optimal arms. The proofs cover arbitrary finite instances, assume neither monotonicity nor termination at instances with two best arms, and require no prior estimates of the gaps, complexity, or gap entropy.

\subsection*{AI use statement}
The main mathematical derivations and the Lean~4 formalization and
verification work were carried out primarily by \texttt{gpt-6-astra}.
The resulting mathematical arguments and formal verification work were
subsequently checked by \texttt{claude-fable-5.1}. Generative AI tools were
also used extensively in preparing the manuscript. These model-based
checks are distinct from the kernel checks described in
Appendix~\ref{verif:section}.
\label{page:main-end}

\appendix
\section{Probabilistic and elimination preliminaries}\label{sec:prelim}
We collect elementary ingredients used by both upper bounds and by the lower-bound comparison. All sample blocks below are fresh. Conditional statements therefore remain valid for an adaptively chosen active set and tolerance.

\subsection{Gaussian tails and finite adaptive experiments}
For the mean $\overline X_m$ of $m$ independent $\mathcal N(\mu,1)$ observations,
\begin{equation}\label{eq:gaussian-tail}
 \P\bigl(\overline X_m-\mu\ge t\bigr)\le e^{-mt^2/2},\qquad
 \P\Bigl(\bigl|\overline X_m-\mu\bigr|>t\Bigr)\le2e^{-mt^2/2}\quad(t>0).
\end{equation}
We write $\kl(p,q)=p\log(p/q)+(1-p)\log\bigl((1-p)/(1-q)\bigr)$ for binary relative entropy, using the usual extended-real conventions.

The following finite-horizon calculation is the Gaussian form of the
adaptive change-of-measure argument used in bandit lower bounds
\citep{kaufmann2016complexity}. We include the proof to account
explicitly for internal randomness and latent mixtures.
\begin{lemma}[Change of distribution for a finite adaptive Gaussian experiment]\label{lem:gaussian-kl}
Run the same algorithm on the mean vectors $\mu$ and $\nu$, stopping its transcript after return or after $m$ total observations. If $P_\mu^{(m)}$ and $P_\nu^{(m)}$ are the transcript laws and $N_i^{(m)}=N_i(m)$ is the number of observations of arm $i$ in this transcript, using the stopped-count convention of Section~\ref{sec:model}, then
\begin{equation}\label{eq:adaptive-kl}
 \KL\bigl(P_\mu^{(m)}\bigm\Vert P_\nu^{(m)}\bigr)
 \le\frac12\sum_i(\mu_i-\nu_i)^2\E_\mu N_i^{(m)}.
\end{equation}
Equality holds when the algorithm's internal random seed is included in the transcript. The same upper bound, averaged over a common latent-variable prior, holds for mixtures of such experiments after the latent variable is discarded.
\end{lemma}
\begin{proof}
Condition on the seed. Given a past transcript, the next selected arm is the same function of that history under both laws. Its conditional Gaussian divergence is $(\mu_i-\nu_i)^2/2$. Sampling actions and return decisions contribute no further divergence. The relative-entropy chain rule, applied for at most $m$ observations with an absorbing state after return, proves the identity for the augmented transcript. Discarding the seed gives the inequality by data processing. For a common latent prior, apply the same chain rule after exposing the latent variable and then discard it.
\end{proof}

For any event $F$, data processing also gives $\kl\bigl(P(F),Q(F)\bigr)\le\KL(P\Vert Q)$. We use Pinsker's inequality $\bigl|P(F)-Q(F)\bigr|\le\sqrt{\KL(P\Vert Q)/2}$ and the elementary bound
\begin{equation}\label{eq:binary-event-bound}
 \kl(p,q)\ge p\log(1/q)-\log2.
\end{equation}
These comparisons will always be applied at finite horizons before taking increasing limits of return events.

\subsection{A median-elimination subroutine with a deterministic sample budget}
This subroutine follows the median-elimination principle of
\citet{even2002pac,even2006action}; the formulation below records the
deterministic budget and Gaussian guarantee needed here.
\begin{lemma}[Median elimination]\label{lem:pac}
For a set of $s$ arms, $0<\varepsilon\le1$, and $0<\beta<0.1$, a procedure $\mathsf{PAC}(S,\varepsilon,\beta)$ returns an arm of mean at least $\max_{i\in S}\mu_i-\varepsilon$ with probability at least $1-\beta$. Its number of samples, denoted $N_{\mathrm{PAC}}(s,\varepsilon,\beta)$, is deterministic and satisfies
\[
 N_{\mathrm{PAC}}(s,\varepsilon,\beta)
 \le C_{\mathrm{PAC}}s\varepsilon^{-2}\log(8/\beta).
\]
A singleton is returned without sampling.
\end{lemma}
\begin{proof}
Start with $S_0=S$. For $r=0,1,\ldots$, while $|S_r|\ge2$, put
\[
 \varepsilon_r=\frac{\varepsilon}{8}(3/4)^r,\qquad
 \beta_r=\beta2^{-(r+1)},\qquad
 m_r=\ceil[\big]{8\varepsilon_r^{-2}\log(8/\beta_r)}.
\]
Sample each arm of $S_r$ exactly $m_r$ times and retain the empirical upper $\ceil[\big]{|S_r|/2}$ arms. These retained sizes and the stopping round are deterministic. By \eqref{eq:gaussian-tail}, the probability that a specified empirical mean exceeds its own true mean by more than $\varepsilon_r$, or falls below it by that amount, is at most $(\beta_r/8)^4$.

Except with probability at most $3(\beta_r/8)^4\le\beta_r$, the current true best is not underestimated by $\varepsilon_r$, and fewer than half the current arms overestimate by $\varepsilon_r$. The second assertion follows by Markov's inequality. If that best arm is discarded, every retained estimate is at least its estimate, and at least one retained arm has no upward error. The largest retained true mean is consequently within $2\varepsilon_r$ of the previous maximum. Summing the losses gives $2\sum_r\varepsilon_r=\varepsilon$, and summing the failure probabilities gives $\sum_r\beta_r=\beta$.

Before termination, $|S_r|\le2s2^{-r}$. Thus the sample sum is bounded by a constant times
\[
 s\varepsilon^{-2}\sum_{r\ge0}(8/9)^r
       \bigl[\log(8/\beta)+r+1\bigr]
 =O\bigl(s\varepsilon^{-2}\log(8/\beta)\bigr).
\]
The integer ceilings are absorbed by the same bound. This also specifies the exact deterministic sample count used in later reservations.
\end{proof}

\subsection{Threshold elimination and retention of at least half the arms}
Given $S$, tolerance $0<d\le1$, and $0<\alpha<0.1$, choose
$a=\mathsf{PAC}(S,d/8,\alpha/16)$. Independently estimate its mean with
$\ceil[\big]{512d^{-2}\log(32/\alpha)}$ observations, obtaining $z$. Independently estimate every active arm with exactly
\begin{equation}\label{eq:arm-sample-count}
 m(S,d,\alpha)=\ceil[\big]{512d^{-2}\log(128/\alpha)}
\end{equation}
observations per arm, obtaining $x_i$. Define the threshold retained set
\begin{equation}\label{eq:raw-threshold}
 R_0=\{i\in S:x_i\ge z-d/2\}.
\end{equation}
The raw threshold procedure returns $R_0$. The modified procedure returns a set $R$ consisting of $R_0$ together with the $\ceil[\big]{|S|/2}$ arms with largest empirical means. Equivalently, add arms in decreasing empirical-mean order until at least half are retained. This equivalence holds because a threshold set is itself an upper segment of the empirical ranking. Use a fixed tie rule on the probability-zero event of equal fresh empirical means. Both procedures use the same observations and have the same sample cost.

\begin{lemma}[Elimination guarantees]\label{lem:elimination}
The raw threshold procedure above has deterministic cost at most
$C_{\mathrm{elim}}|S|d^{-2}\log(128/\alpha)$. It satisfies:
\begin{enumerate}
\item If a unique best arm belongs to $S$, it belongs to $R_0$ except with probability at most $\alpha$.
\item If the unique best arm belongs to $S$, at most $N$ active arms (including that best) have gap less than $d$ from it, and $s=|S|>4N$, then $R_0$ both contains the best and has size at most $\ceil[\big]{s/2}$ except with probability at most $\alpha$.
\item If $s\le4$, the unique best arm belongs to $S$, and every other active arm has gap at least $d$ from it, then $R_0$ consists of exactly the best arm except with probability at most $\alpha$.
\end{enumerate}
Replacing $R_0$ by the enlarged set $R$ preserves the first two guarantees. Under the good event used to prove the third guarantee, $R$ retains the best and has size $\ceil[\big]{s/2}$; at most two such good calls finish from size at most four. The singleton conclusion in the third guarantee applies to $R_0$, not necessarily to $R$.
\end{lemma}
\begin{proof}
Let $\mu_S=\max_{i\in S}\mu_i$. Except on the PAC failure and the reference-estimation failure, each of probability at most $\alpha/16$,
\begin{equation}\label{eq:reference-interval}
 \mu_S-3d/16\le z\le \mu_S+d/16.
\end{equation}
Each active estimate has two-sided error exceeding $d/16$ with probability at most $\alpha/64$. The best survives whenever its estimate is at least $\mu_S-d/16$. An arm with gap at least $d$ survives threshold elimination only if it overestimates by at least $5d/16$, an event contained in its $d/16$ deviation event. If $s>4N$, the expected number of such far-arm deviations is at most $s\alpha/64$, so the probability of at least $s/4$ of them is at most $\alpha/16$. Off that event, fewer than $s/4$ near arms and fewer than $s/4$ far arms survive. Adding the reference and best-estimation failures proves the second assertion with slack in the constants. The first assertion uses only retention of the best arm. For $s\le4$, a union bound over all current estimates gives the third assertion. With all these estimates accurate and all gaps at least $d$, the best also has the largest empirical mean, so it remains in the retained empirical upper half. Lemma~\ref{lem:pac} and the explicit sample counts prove the cost bound.
\end{proof}

\begin{remark}[Reusing a reference estimate]\label{rem:reference-reuse}
Suppose a reference is selected at scale entry with error budget $\alpha_0/16$ and its mean estimate is separately accurate to $d/16$ with failure probability $\beta$. If the original best is still active, \eqref{eq:reference-interval} continues to hold after all later same-scale deletions, even when the reference label itself has been removed. Later calls require only fresh active-arm estimates, and their conditional estimation/progress failures are charged to their own $\alpha$ budgets. The PAC and reference failures are charged once at entry. For correctness after loss of the best, only the one-sided inequality $z\le\mu_{[1]}+d/16$ is needed. In particular, a reference satisfying the upper bound and $\Delta_i\le d/8$ imply that a removed arm satisfies $x_i-\mu_i<-5d/16$. Selecting a reference arm with a low mean does not invalidate this implication. The single algorithm of Appendix~\ref{alg:section} estimates the reference once per scale with its own failure budget $\beta_{j,k}$ from \eqref{alg:reference-samples} instead of $\alpha/16$; the per-call bound of Lemma~\ref{lem:elimination} is therefore not applied verbatim there, and the corresponding failure events are accounted for separately in Appendices~\ref{sound:section} and~\ref{cost:section}.
\end{remark}

\subsection{Geometrically weighted group sizes and their entropy}
For an instance, put $n_k=|G_k|$ and $K=\max\{k:n_k>0\}$. For $0\le k\le K$, define
\begin{equation}\label{eq:smoothed-weights}
 d_k=2^{-k},\quad N_k=1+\#\bigl\{i\ne i_{[1]}:\Delta_i<d_k\bigr\},\quad
 W_k=4^k\biggl(1+\sum_{t\ge k}n_t\biggr),\quad W=\sum_{k=0}^K W_k.
\end{equation}
Here $W_0=n$, $N_K=1$, and $W_k=4^kN_{k-1}$ for $k\ge1$.
\begin{lemma}[Entropy of the normalized weights]\label{lem:smoothing}
Let $q_k=W_k/W$. Then
\begin{equation}\label{eq:smoothing-bounds}
 H\le W<\frac{32}{3}H,\qquad
 \Ent(q)\le\frac{32}{3}\Ent(I)+C_0,\qquad 4^K\le4D,
\end{equation}
where $C_0=\log(4/3)+(4/9)\log4$ is an absolute constant.
\end{lemma}
\begin{proof}
Put $v_t=\bigl(n_t+\1\{t=K\}\bigr)4^t$ and $s_t=\sum_{l=0}^t4^{-l}$. The extra term adds one fictitious arm at the finest scale. Directly,
\[
 W_k=\sum_{t\ge k}v_t4^{k-t},\qquad
 W=\sum_t v_ts_t,\qquad 1\le s_t\le4/3.
\]
Because $H_t\le n_t4^t<4H_t$ on every nonempty group and $4^K<4D\le4H$, the total $V=\sum_t v_t$ satisfies $H\le V<8H$. This proves the bounds on $W$; the bound on $4^K$ follows from the group with largest index.

Sample $J$ with probabilities $r_t=v_ts_t/W$. Conditional on $J=t$, sample $Z\in\{0,\ldots,t\}$ with probabilities $4^{-l}/s_t$. Then $J-Z$ has distribution $q$, and
\[
 \Ent(Z\mid J=t)=\log s_t+\log4\,\E[Z\mid J=t]
 \le\log(4/3)+(4/9)\log4=C_0.
\]
Since $n_K\ge1$, one has $v_t<8H_t$ on every nonempty group, and therefore $r_t\le(32/3)p_t$. Nonnegativity of relative entropy gives
\[
 \Ent(r)=\sum_t r_t\log(1/p_t)-\KL(r\Vert p)
 \le\frac{32}{3}\sum_t p_t\log(1/p_t).
\]
Finally, entropy cannot increase upon discarding information, so
\[
 \Ent(q)=\Ent(J-Z)\le\Ent(J,Z)\le\Ent(r)+C_0.
\]
\end{proof}

\section{A gap-entropy lower bound on the original instance}
\label{lower:section}

We prove the lower-bound direction of Theorem~\ref{thm:entropy} for
every $\delta$-correct algorithm of Section~\ref{sec:model};
Appendix~\ref{bridge:section} extends it to algorithms with arbitrary
standard Borel internal randomness and to history-dependent probability kernels. The argument compares several
labelings of the original instance with one experiment with two best arms. Every
expected number of samples in the comparison is evaluated on the original
instance.

\begin{theorem}[Lower bound]
\label{lower:theorem}
For every instance $I$ with $n\ge2$ and every $0<\delta<1/10$,
\begin{equation}
\label{lower:main}
 \cL(I,\delta)\ge H(I)\max\bigl\{2\kl(1-\delta,\delta),\Ent(I)/4\bigr\}
 \ge\frac{H(I)}5\bigl(\log(1/\delta)+\Ent(I)\bigr).
\end{equation}
\end{theorem}

\subsection{Symmetrization and finite transcript experiments}

Fix a $\delta$-correct algorithm. Randomly permute its input
labels, run the algorithm through this relabeling, and undo the
permutation on its answer. The resulting algorithm $A$ is permutation
equivariant in distribution, remains $\delta$-correct, and its
expected stopping time on $I$ is exactly the original algorithm's
permutation-averaged expected stopping time. It therefore suffices to
prove~\eqref{lower:main} for $A$ on $I$. If its expected stopping time
is infinite, there is nothing to prove. Assume henceforth that it is
finite.

Let $\mathcal D$ be the finite set of distinct positive gaps in $I$,
let $m_d$ be the multiplicity of $d\in\mathcal D$, and write
$\mu_{[1]}$ for the best mean. Permutation equivariance implies that
the expected number of samples of an arm of gap $d$ is the same on every
labeling of the target multiset. Denote this count by $t_d$, and set
\begin{equation}
\label{lower:normalized-count}
 \alpha_d=d^2t_d.
\end{equation}
In particular, conditioning a uniformly labeled target instance to give
one designated label a specified gap does not change the expected count
of any other specified arm with a given gap. We have
\begin{equation}
\label{lower:target-cost}
 \frac{\E_I\tau_A}{H}
 \ge \sum_{d\in\mathcal D}\frac{m_d d^{-2}}H\alpha_d.
\end{equation}
The inequality only discards samples from the best arm.

We use finite experiments throughout changes of distribution. A transcript
at horizon $m\in\N$ records the realization of the algorithm's internal random seed, chosen arms,
observations, and any returned answer through at most $m$ samples; after
the algorithm returns, the transcript is extended by an absorbing
symbol. The chain rule for relative entropy gives, for two fixed mean
vectors $\mu,\nu$ and the same algorithm,
\begin{equation}
\label{lower:finite-kl}
 \KL\bigl(P_\mu^{(m)}\bigm\Vert P_\nu^{(m)}\bigr)
 =\frac12\sum_{i=1}^n(\mu_i-\nu_i)^2
       \E_\mu N_i(\tau_A\wedge m).
\end{equation}
Here $N_i(\tau_A\wedge m)$ counts observations from arm $i$ in the
finite experiment. To see the identity, condition successively on each
past transcript. The algorithm's conditional distribution of its next
action is the same under the two experiments, so it contributes zero
relative entropy. An observation from arm $i$ contributes
$\KL\bigl(\mathcal N(\mu_i,1)\bigm\Vert\mathcal N(\nu_i,1)\bigr)
=(\mu_i-\nu_i)^2/2$. Summing these conditional contributions proves
the identity, including arbitrary dependence on the internal random seed. The same
argument applies when the experiment is censored by an additional
stopping rule. We never require the algorithm to terminate on an
instance with tied best arms.

\subsection{Distributions on the original instance and instances with two best arms}

For $d\in\mathcal D$, let $P_d$ be the law obtained by placing a gap-$d$
arm at label $1$, uniformly assigning the remaining target multiset to
labels $2,\ldots,n$, and running $A$. This label is designated only in
the analysis and is not identified to the algorithm. If $N_1$ is the
number of samples from label $1$ before a return, then
\begin{equation}
\label{lower:source-facts}
 \E_{P_d}N_1=t_d,
 \qquad
 P_d\bigl(\tau_A<\infty,\widehat{\imath}\ne1\bigr)\ge1-\delta.
\end{equation}
For $a\in\mathcal D$, define $Q_a$ by assigning mean $\mu_{[1]}$ to
label $1$ and uniformly placing, at the remaining labels, the multiset
of $I$ with one gap-$a$ arm removed. This comparison instance has
exactly two arms of mean $\mu_{[1]}$.

\begin{lemma}[Output probabilities when two arms are best]
\label{lower:null}
For every $a\in\mathcal D$,
\begin{equation}
\label{lower:null-bound}
 Q_a\bigl(\tau_A<\infty,\widehat{\imath}\ne1\bigr)\le\delta.
\end{equation}
\end{lemma}
\begin{proof}
Fix a labeled mean vector occurring in $Q_a$. Lower its other
mean-$\mu_{[1]}$ arm to $\mu_{[1]}-\varepsilon$, leaving label $1$
unchanged. Because $I$ has at least two arms, a unique best arm, and means
in $[0,1]$, $\mu_{[1]}>0$. Thus every sufficiently small
$\varepsilon>0$ gives a valid instance in which label $1$ is uniquely
best. This perturbation is valid also at the endpoint $\mu_{[1]}=1$.
Correctness bounds the probability of returning another label
by horizon $m$ by $\delta$ on the perturbed instance.

By~\eqref{lower:finite-kl}, the relative entropy between the unperturbed
and perturbed horizon-$m$ experiments is at most $m\varepsilon^2/2$.
Pinsker's inequality therefore makes their total variation tend to zero
as $\varepsilon\downarrow0$. The same error bound holds on the fixed
instance with two best arms at every finite horizon. Average over the finitely many
background labelings and then let $m\to\infty$. The finite-return
events increase to the event in~\eqref{lower:null-bound}, proving the
claim without a termination assumption on $Q_a$.
\end{proof}

\begin{lemma}[Changing two arm means]
\label{lower:coupling}
For every finite horizon $m$ and $a,d\in\mathcal D$ with $a\ne d$,
\begin{equation}
\label{lower:two-coordinate}
 \KL\bigl(P_d^{(m)}\bigm\Vert Q_a^{(m)}\bigr)
 \le\frac12\bigl(d^2t_d+(a-d)^2t_a\bigr).
\end{equation}
For $a=d$, the upper bound is $\alpha_d/2$. Consequently, if $a\ge d$,
\begin{equation}
\label{lower:coarser-comparison}
 \KL\bigl(P_d^{(m)}\bigm\Vert Q_a^{(m)}\bigr)
 \le\frac12(\alpha_d+\alpha_a).
\end{equation}
\end{lemma}
\begin{proof}
Suppose first that $a\ne d$. Draw a labeled assignment from the source
experiment $P_d$, and independently of all observations and of the
algorithm's internal random seed, choose uniformly one of its $m_a$ gap-$a$ arms. Raise label
$1$ from $\mu_{[1]}-d$ to $\mu_{[1]}$, and change the selected arm
from $\mu_{[1]}-a$ to $\mu_{[1]}-d$. The resulting background
multiset is the original target multiset with one gap-$a$ arm removed.

We verify its uniform labeling, since a common comparison law is
essential. Let $M_P$ and $M_Q$ count the distinct background labelings
in $P_d$ and $Q_a$. Their multiplicities give
$M_Q/M_P=m_a/m_d$. Any fixed resulting background assignment has
exactly $m_d$ preimages among pairs of a source assignment and a
selected label:
choose which of its $m_d$ gap-$d$ labels was changed from gap $a$.
Each pair has probability $1/(M_Pm_a)$, so this background has
probability $m_d/(M_Pm_a)=1/M_Q$. Thus the induced law is exactly
$Q_a$, including all repeated-gap cases.

Keep the source assignment and selected label as latent variables with
the same prior in both experiments. Conditional on them, the means
differ at precisely the two described coordinates. Under the source
law, their expected numbers of samples are $t_d$ and $t_a$, respectively:
the first follows from symmetry, and the second follows from the same
symmetry and the independent uniform selection. More explicitly, for
every fixed source assignment and every eligible selected label, the
source transcript kernel is just $A$ on that assignment, and its
expected count at that label is $t_a$. The hidden selection does not
change this kernel. In the comparison experiment, each latent pair
likewise uses $A$ on its mapped mean vector; pairs with the same mapped
vector give the same conditional transcript law. Apply
\eqref{lower:finite-kl} conditionally, average over the common latent
prior, and bound truncated counts by their full source expectations.
This gives the right-hand side of~\eqref{lower:two-coordinate} for the
joint laws including the latent variables. Discarding latent variables
cannot increase relative entropy, proving the asserted marginal
inequality. Neither latent variable is revealed to the algorithm.

When $a=d$, raising label $1$ alone already induces $Q_d$, and its KL divergence is at most $d^2t_d/2$. Finally, if $a\ge d$, then
$(a-d)^2\le a^2$, which proves~\eqref{lower:coarser-comparison}.
All numbers of samples used here belong to labelings of $I$.
\end{proof}

The $a=d$ case, binary data processing for the event of returning a label
other than $1$ by horizon $m$, and then the limit $m\to\infty$ yield
\begin{equation}
\label{lower:baseline}
 \alpha_d\ge2\kl(1-\delta,\delta)
 \ge\log(1/\delta)>1.
\end{equation}
Indeed, let $F_m=\bigl\{\tau_A\le m,\ \widehat{\imath}\ne1\bigr\}$ be the event of
returning a label other than $1$ by horizon $m$, and put
$p_m=P_d^{(m)}(F_m)$ and $q_m=Q_d^{(m)}(F_m)$. Binary data processing
gives $\kl(p_m,q_m)\le\alpha_d/2$ at every finite horizon. By
Lemma~\ref{lower:null}, $q_m\le\delta$, and by
\eqref{lower:source-facts}, $p_m$ increases to
$P_d\bigl(\tau_A<\infty,\widehat{\imath}\ne1\bigr)\ge1-\delta>\delta$. Fix
$\varepsilon>0$ with $1-\delta-\varepsilon>\delta$. For all
sufficiently large $m$, $p_m\ge1-\delta-\varepsilon>\delta\ge q_m$.
Since $q\mapsto\kl(p,q)$ is decreasing on $(0,p]$, with
$\kl(p,0)=+\infty$ for $p>0$, and
$p\mapsto\kl(p,\delta)$ is increasing on $[\delta,1]$,
\[
 \kl(1-\delta-\varepsilon,\delta)\le\kl(p_m,\delta)
 \le\kl(p_m,q_m)\le\alpha_d/2
\]
for all such $m$. Letting $\varepsilon\downarrow0$ and using the
continuity of $\kl(\cdot,\delta)$ on $[0,1)$ gives
$\kl(1-\delta,\delta)\le\alpha_d/2$.
For the elementary second inequality, write
$\kl(1-\delta,\delta)=(1-2\delta)
\log\bigl((1-\delta)/\delta\bigr)$. If $L=\log(1/\delta)\ge\log10$,
then $2\kl(1-\delta,\delta)\ge1.6(L+\log0.9)\ge L$.

\subsection{Disjoint stopping events under a common comparison distribution}

Set $c_0=1/100$ and define events on returned transcripts by
\begin{equation}
\label{lower:window}
 E_d=\bigl\{\tau_A<\infty,\ \widehat{\imath}\ne1,\
          c_0d^{-2}\le N_1\le16t_d\bigr\}.
\end{equation}
These are transcript events whose endpoints depend only on the fixed
target and the algorithm, not on the law under which they are evaluated.

\begin{lemma}[Probabilities of the stopping events]
\label{lower:windows}
For every $d\in\mathcal D$, $P_d(E_d)>1/2$. Whenever $a\ge d$,
\begin{equation}
\label{lower:null-window}
 Q_a(E_d)\ge\frac14\exp(-\alpha_d-\alpha_a).
\end{equation}
\end{lemma}
\begin{proof}
By~\eqref{lower:baseline}, $\alpha_d>1$, so $t_d=\alpha_dd^{-2}>0$ and
the interval $\bigl[c_0d^{-2},16t_d\bigr]$ is nonempty; in particular Markov's
inequality applies and gives $P_d(N_1>16t_d)\le1/16$.
For the lower endpoint, compare $P_d$ with $Q_d$ and censor both
experiments immediately before a request that would take the count of
label $1$ above
$\ceil[\big]{c_0d^{-2}}-1$, the largest integer strictly below $c_0d^{-2}$.
Also impose a total-sample horizon $m$. Every return by this horizon
with $N_1<c_0d^{-2}$ is preserved by this censoring. The KL divergence is at
most $c_0/2$, since only label $1$ changes. Pinsker's inequality gives
total variation at most $\sqrt{c_0}/2=0.05$. The comparison probability
of such an early return with answer different from $1$ is at most
$\delta$ by Lemma~\ref{lower:null}. Letting $m\to\infty$ gives
\[
 P_d\bigl(\tau_A<\infty,\widehat{\imath}\ne1,N_1<c_0d^{-2}\bigr)
 \le\delta+0.05.
\]
Together with~\eqref{lower:source-facts}, this implies
\[
 P_d(E_d)\ge1-2\delta-0.05-1/16>1/2.
\]

Apply~\eqref{lower:coarser-comparison} and binary data processing to
$E_d\cap\{\tau_A\le m\}$, and take the increasing-event limit. If
$p=P_d(E_d)>1/2$ and $q=Q_a(E_d)$, then
\[
 \frac12(\alpha_d+\alpha_a)
 \ge\kl(p,q)\ge p\log(1/q)-\log2.
\]
The last inequality follows by dropping the nonnegative term
$(1-p)\log\bigl(1/(1-q)\bigr)$ and using that binary entropy is at most $\log2$.
Since $p>1/2$, rearranging proves~\eqref{lower:null-window}. Finite
transcript bounds justify this step even if $A$ never returns on a set
of positive $Q_a$-probability.
\end{proof}

In each nonempty group choose an exact gap $d_k$ minimizing
$\alpha_d$, and abbreviate $\alpha_k=\alpha_{d_k}$. For $x\ge1$ let
\[
 \mathcal C_x=\{k:\alpha_k\le x\},\qquad N(x)=|\mathcal C_x|.
\]
If $\mathcal C_x$ is nonempty, choose its largest-gap representative
$a$. Then $a\ge d_k$ and $\alpha_a\le x$ for every
$k\in\mathcal C_x$. Lemma~\ref{lower:windows} gives the simultaneous
bounds under this \emph{single} law $Q_a$,
\begin{equation}
\label{lower:cheap-window}
 Q_a(E_{d_k})\ge\tfrac14e^{-2x},\qquad k\in\mathcal C_x.
\end{equation}

On the logarithmic sample-count coordinate $\log_4(N_1/c_0)$, the
event for $d_k$ corresponds to the interval
\[
 [s_k,s_k+b_k],\qquad
 s_k=\log_4\bigl(d_k^{-2}\bigr)\in(k-1,k],\qquad
 b_k=\log_4(1600\alpha_k).
\]
All intervals with $\alpha_k\le x$ have length at most $b(x)=\log_4(1600x)$.
There is at most one starting point in each cell $(k-1,k]$, so any
interval of length $u$ contains at most $u+2$ starting points. Greedily
select the remaining interval with the smallest starting point, then
delete every remaining interval whose start is at most its right
endpoint. Each selection deletes at most $b(x)+2$ intervals, and
the selected closed intervals are pairwise disjoint. We obtain at least
$N(x)/\bigl(b(x)+2\bigr)$ disjoint events, all contained in
$\bigl\{\tau_A<\infty,\widehat{\imath}\ne1\bigr\}$. Equations
\eqref{lower:null-bound} and~\eqref{lower:cheap-window} imply
\begin{equation}
\label{lower:counting}
 N(x)\le4\delta\,\bigl[\log_4(1600x)+2\bigr]e^{2x}.
\end{equation}
If $\mathcal C_x$ is empty the inequality is immediate. The reference
gap $a$ can depend on $x$: each application proves a deterministic
inequality for the same counting function $N$.

\subsection{From scale counting to entropy}

By~\eqref{lower:baseline}, $N(x)=0$ for $x<1$. Tonelli's theorem and
\eqref{lower:counting} yield
\begin{align}
\label{lower:kraft}
 Z:=\sum_k e^{-4\alpha_k}
 &=\int_0^\infty4e^{-4x}N(x)\,dx\notag\\
 &\le16\delta\int_1^\infty
       \bigl[\log_4(1600x)+2\bigr]e^{-2x}\,dx
 <10\delta<1.
\end{align}
For an explicit verification of the inessential numerical constant,
use $\log x\le x$ for $x\ge1$. The integral multiplied by $16$ is at
most
\[
 16e^{-2}\biggl[
   \frac12\Bigl(\frac{\log1600}{\log4}+2\Bigr)
   +\frac3{4\log4}\biggr]<10.
\]
Let $p_k=H_k/H$ on the nonempty groups. Nonnegativity of relative
entropy between $(p_k)$ and $\bigl(e^{-4\alpha_k}/Z\bigr)$ gives
\begin{equation}
\label{lower:gibbs}
 4\sum_kp_k\alpha_k
 \ge\Ent(I)-\log Z\ge\Ent(I).
\end{equation}
The minimizing representatives and~\eqref{lower:target-cost} imply
\[
 \frac{\E_I\tau_A}H
 \ge\sum_{d\in\mathcal D}\frac{m_d d^{-2}}H\alpha_d
 \ge\sum_kp_k\alpha_k.
\]
By~\eqref{lower:baseline} every $\alpha_k$ is at least
$2\kl(1-\delta,\delta)\ge L=\log(1/\delta)$, and so is this last sum,
since the $p_k$ are probability weights. Hence
\[
 \frac{\E_I\tau_A}H
 \ge\max\bigl\{2\kl(1-\delta,\delta),\Ent(I)/4\bigr\}
 \ge\max\bigl\{L,\Ent(I)/4\bigr\}
 \ge\frac{L+\Ent(I)}5,
\]
where the last step bounds the maximum below by the weighted mean
$\tfrac15L+\tfrac45\cdot\tfrac{\Ent(I)}4$.
Undoing the initial symmetrization and taking the infimum over $\delta$-correct algorithms proves Theorem~\ref{lower:theorem}.
The argument covers repeated suboptimal means, arbitrary positive
gaps, the full range $0<\delta<1/10$, and the endpoint best mean $1$.

\section{An instance-wise upper bound}
\label{pointwise:section}

The upper-bound direction of Theorem~\ref{thm:entropy} concerns a
pointwise infimum over algorithms. We may therefore choose an algorithm
whose fixed parameters depend on the target instance, provided it is
correct on every admissible input. The following construction makes this
quantifier distinction explicit.

\begin{proposition}[Instance-wise upper bound]
\label{pointwise:theorem}
For every target instance $I$ and every $0<\delta<1/10$, there exists a
$\delta$-correct algorithm $A_{I,\delta}$ with finite expected
stopping time on every instance with a unique best arm such that, uniformly over
all labelings $\pi I$ of the target,
\begin{equation}
\label{pointwise:main}
 \E_{\pi I}\tau_{A_{I,\delta}}
 \le C H(I)\bigl(\log(1/\delta)+\Ent(I)\bigr).
\end{equation}
Consequently $\cL(I,\delta)$ is bounded above by the right-hand side.
\end{proposition}

\subsection{Parameters and one finite run}

Fix a target instance; it has $n\ge2$ arms by definition. Its group sizes $n_k=|G_k|$, the largest index $K$ of a nonempty group, and the following quantities are fixed parameters of the construction:
\begin{align}
\label{pointwise:weights}
 d_k&=2^{-k},&
 N_k&=1+\sum_{j>k}n_j,
 &W_k&=4^k\biggl(1+\sum_{j\ge k}n_j\biggr),\notag\\
 W&=\sum_{k=0}^K W_k,&q_k&=W_k/W,
 &&0\le k\le K.
\end{align}
Here $N_k$ is exactly the number of original arms whose gaps are below
$d_k$, including $i_{[1]}$. In particular,
\[
 W_0=n,\qquad W_k=4^kN_{k-1}\ (k\ge1),\qquad N_K=1.
\]
Lemma~\ref{lem:smoothing} gives
\begin{equation}
\label{pointwise:smoothing}
 H\le W<\tfrac{32}{3}H,
 \qquad \Ent(q)=O\bigl(\Ent(I)+1\bigr),
 \qquad 4^K\le4D\le4H.
\end{equation}
Only the group sizes and the quantities derived from them are supplied as constants
to this construction; its correctness will not assume that an input
has these group sizes.

For a run error parameter $0<\varepsilon<1/10$, set
$\eta_k=(\varepsilon/2)q_k$. Start with all input arms active. At
scales $k=0,\ldots,K$, repeatedly apply the raw threshold
elimination procedure of Lemma~\ref{lem:elimination} while
$|S|>4N_k$. On its $r$th invocation at this scale, starting with $r=0$,
use tolerance $d_k$ and confidence
\begin{equation}
\label{pointwise:allocation}
 \alpha_{k,r}=\eta_k2^{-(r+1)}.
\end{equation}
If the returned set $R_0$ is empty or has size greater than
$\ceil[\big]{|S|/2}$, abort the run. Otherwise replace $S$ by $R_0$ and
continue. After all scales, the active set has at most $4N_K=4$ arms.
If it is a singleton, return its member. Otherwise make one final
raw threshold-elimination call at tolerance $d_K$ and confidence $\varepsilon/2$;
return its arm only if $R_0$ is a singleton, and abort otherwise.
Every call uses fresh samples; the construction uses no internal randomness, and ties are broken by a fixed rule.

A run is always finite. Every accepted loop invocation reduces
the active size geometrically, there are finitely many scales, and an
invocation violating the size condition aborts immediately.

\begin{lemma}[Guarantees for one run]
\label{pointwise:attempt}
The preceding run has the following properties.
\begin{enumerate}
 \item On every input with $n$ arms and a unique best arm, its probability of an
 incorrect return is at most $\varepsilon$.
 \item On every labeling of the target $I$, it returns the correct arm
 with probability at least $1-\varepsilon$.
 \item On every input with $n$ arms and every trajectory, its sample count is at most
 \begin{equation}
 \label{pointwise:attempt-budget}
 B_I(\varepsilon)
 =C H(I)\bigl(\log(1/\varepsilon)+\Ent(I)\bigr)
 \end{equation}
 for a sufficiently large universal constant $C$.
\end{enumerate}
\end{lemma}
\begin{proof}
For the first claim, an incorrect singleton can be returned only after
the input's original best has been removed. Conditional on its survival
to a call, the best-arm retention guarantee in Lemma~\ref{lem:elimination} bounds
this removal probability by that call's confidence. The confidence sum
over every possible trajectory is at most
\[
 \sum_{k=0}^K\sum_{r\ge0}\alpha_{k,r}+\varepsilon/2
 =\sum_{k=0}^K\eta_k+\varepsilon/2=\varepsilon.
\]
An abort cannot create an incorrect answer. This argument requires no agreement between the input and the fixed
target instance.

For the second claim, on a labeling of $I$ the active number of arms
within gap $d_k$ of the original best is at most $N_k$, as long as that
best remains active. Thus every loop invocation meets the halving
condition in Lemma~\ref{lem:elimination}. Except with its allocated
error probability, it protects the best and satisfies the acceptance
condition. At scale $K$ all original suboptimal gaps are at least
$d_K$. The final invocation therefore satisfies the primitive's
small-set singleton guarantee. A union bound with the same confidence
sum proves target success with probability at least $1-\varepsilon$.

For the deterministic cost claim, the active size entering scale zero
is $n=W_0$. If scale $k\ge1$ is reached, its entry size is at most
$4N_{k-1}$ by the preceding loop's stopping condition, independently
of all statistical events. Every loop call starts with at least five
arms; an accepted result has at most $\ceil[\big]{s/2}\le3s/4$ arms.
Consequently, including a possible final aborting invocation, the
scale's cost is at most
\begin{align*}
 C W_k\sum_{r\ge0}(3/4)^r
   \bigl(\log(1/\eta_k)+(r+1)\log2\bigr)
 &\le C'W_k\log(1/\eta_k).
\end{align*}
Here $C$ absorbs $C_{\mathrm{elim}}$ and the additive $\log128$ in the
primitive's cost $C_{\mathrm{elim}}|S|d^{-2}\log(128/\alpha)$, using
$\log(1/\eta_k)\ge\log(20)$.
This uses the deterministic sample budget in
Lemma~\ref{lem:elimination}, not a cost bound conditioned on successful
elimination. The final call uses at most
$C4^K\log(2/\varepsilon)$ samples (the factor $|S|\le4$ and the term
$\log256$ of the primitive's cost at confidence $\varepsilon/2$ are
absorbed into $C$). Summing and using
\eqref{pointwise:smoothing},
\begin{align*}
 B_I(\varepsilon)
 &\le C\Biggl[
   \sum_{k=0}^K W_k\log\frac{2W}{\varepsilon W_k}
        +H\log(2/\varepsilon)\Biggr]\\
 &=C\Bigl[W\bigl(\log(2/\varepsilon)+\Ent(q)\bigr)
        +H\log(2/\varepsilon)\Bigr]\\
 &\le C'H\bigl(\log(1/\varepsilon)+\Ent(I)\bigr).
\end{align*}
The additive entropy constant is absorbed by
$\log(1/\varepsilon)>\log10$. These are bounds in terms of the fixed
target even when the actual input has different means or gaps.
\end{proof}

\subsection{A finite-expectation fallback}

For completeness, we specify a fallback $F$ that is $\delta/2$-correct
and has finite expected number of samples on every input with a unique best arm.
At round $r=0,1,\ldots$, sample each of the $n$ arms cumulatively
$m_r=2^r$ times and let $\widehat\mu_{i,r}$ be its empirical mean.
Use confidence radius
\begin{equation}
\label{pointwise:fallback-radius}
 \rho_r=
 \sqrt{\frac2{m_r}\log\frac{8n(r+1)^2}{\delta}}.
\end{equation}
Stop once one interval
$\bigl[\widehat\mu_{i,r}-\rho_r,\widehat\mu_{i,r}+\rho_r\bigr]$
lies strictly above all the others, and return its arm. Gaussian tails
give
\[
 \P\Bigl(\bigl|\widehat\mu_{i,r}-\mu_i\bigr|>\rho_r\Bigr)
 \le\frac{\delta}{4n(r+1)^2}.
\]
The union bound over all arms and rounds is less than $\delta/2$.
When every interval contains its true mean, a returned arm must be the
unique best, proving correctness.

Fix now any input with smallest gap $\Delta_{[2]}>0$. There exists a
finite $r_0$ such that $\rho_r\le\Delta_{[2]}/8$ for all
$r\ge r_0$. If every empirical error at such a round is at most
$\Delta_{[2]}/8$, the best arm's interval is strictly above every
other interval. Therefore
\[
 \P\bigl(F\text{ has not stopped by round }r\bigr)
 \le2n\exp\bigl(-m_r\Delta_{[2]}^2/128\bigr),\qquad r\ge r_0.
\]
It follows, for example by summing each next round's sample cost times
its reaching probability, that
\[
 \E\tau_F
 \le nm_{r_0}+\sum_{r\ge r_0}
       nm_{r+1}\,2n\exp\bigl(-m_r\Delta_{[2]}^2/128\bigr)<\infty.
\]
No sharp instance-dependent guarantee is required of this fallback.

\subsection{Correctness on every instance and target expected cost}

Perform independent runs, indexed by $t=0,1,\ldots$, with
\[
 \varepsilon_t=\delta2^{-(t+2)}.
\]
An abort starts the next run; an answer stops the stream. In
parallel, run an independent copy of $F$, alternating one requested
sample from each stream, and return whichever stream answers first.
Internal computation between sample requests does not contribute to
sample complexity. On an input with a different arm count from the
fixed target, run $F$ alone. A singleton input returns immediately.
This defines the algorithm on all admissible input sizes.

Lemma~\ref{pointwise:attempt} bounds the run stream's total error
probability by
$\sum_{t\ge0}\varepsilon_t=\delta/2$. Together with the fallback's
bound, the combined algorithm is $\delta$-correct. Its sample
count is at most twice the hypothetical fallback sample count, up to
one sample. Hence it has finite expectation and terminates almost
surely on every input with a unique best arm, including instances different from the fixed target.

On a labeling of $I$, fresh observations and
Lemma~\ref{pointwise:attempt} imply that, conditional on reaching a run, its abort probability is at most $\varepsilon_t$. Thus the
probability that the run-only stream reaches run $t$ is at
most
\[
 \prod_{u<t}\varepsilon_u\le2^{-t}.
\]
Its expected total sample count is consequently at most
\begin{align*}
 \sum_{t\ge0}2^{-t}B_I(\varepsilon_t)
 &\le C H\sum_{t\ge0}2^{-t}
       \bigl(L+\Ent(I)+(t+2)\log2\bigr)\\
 &=O\Bigl(H\bigl[L+\Ent(I)\bigr]\Bigr).
\end{align*}
Interleaving increases this upper bound by at most a factor two, up to
one sample, because an earlier fallback return only shortens the
execution. The constants are uniform over all target labelings. This
proves Proposition~\ref{pointwise:theorem}, and taking the pointwise
infimum proves the upper-bound direction of
Theorem~\ref{thm:entropy}.

\section{A single algorithm}
\label{alg:section}

We construct a single algorithm for Theorem~\ref{thm:universal}. It is
given only the arms and the confidence level; the complexity, gaps, and
gap entropy enter solely into its analysis. A singleton input is returned
without sampling. Throughout the construction, assume $n\ge2$ and set
\begin{equation}
 L=\log(1/\delta),\qquad \gamma=\delta/1024.
 \label{alg:constants}
\end{equation}
The algorithm performs successive runs indexed by $j=0,1,\ldots$. The bounds on their total weight and number of samples are
\begin{equation}
 h_j=2^j,\qquad M_j=1024h_j,\qquad
 Q_j=\ceil[\big]{C_Qh_jL},
 \label{alg:caps}
\end{equation}
where $C_Q$ is the absolute constant $C_*$ of
\eqref{cost:sample-cap-fits}, which depends only on the constants of
Lemmas~\ref{lem:pac} and~\ref{lem:elimination}; any larger absolute
constant also works, and the algorithm is fully specified once it is
fixed.
All observations in a new run are fresh; the algorithm uses no
internal randomness.

\subsection{Sampling and confidence parameters}
\label{alg:parameters}

Each run maintains an active set $S$, a scale $k\ge0$, the cumulative weight of its calls, and its reserved number of samples. A call with $s=|S|\ge2$ uses
\begin{equation}
 d=2^{-k},\qquad u=d^{-2}=4^k,\qquad w=su.
 \label{alg:work}
\end{equation}
Here $w$ is the factor in the sample bound before the confidence logarithm. If adding $w$ respects the bound on the cumulative weight, the confidence parameter is
\begin{equation}
 \alpha=\gamma\frac{w}{M_j}
 \begin{cases}
 (j+1)^{-2},&2\le s\le7,\\
 1,&s\ge8.
 \end{cases}
 \label{alg:alpha}
\end{equation}
The branch is recomputed from the active-set size at every call. Each arm estimate in the call uses exactly
\begin{equation}
 m=\ceil[\Big]{512u\log\frac{128}{\alpha}}
 \label{alg:arm-samples}
\end{equation}
fresh samples.

Only the first call at a scale selects a reference. Apply the deterministic
sample-budget PAC procedure of Lemma~\ref{lem:pac} to the active set, with
accuracy $d/8$ and error probability $\alpha/16$. Denote its output by $a$.
The deterministic number of samples used by this procedure is written
$N_{\mathrm{PAC}}(s,d/8,\alpha/16)$. Independently estimate $\mu_a$ using
\begin{equation}
 \beta_{j,k}=\frac{\delta}{64(j+1)^2(k+1)^2},\qquad
 m_z=\ceil[\Big]{512u\log\frac{2}{\beta_{j,k}}}
 \label{alg:reference-samples}
\end{equation}
samples, and denote this empirical mean by $z$. All subsequent calls at
this scale reuse the numerical value $z$. Reuse is allowed even if $a$
has left the active set: no further observation from $a$ is then taken.
The PAC procedure, the estimate $z$, and the estimates used for
elimination use mutually independent fresh data.

\subsection{The procedure and its reservation order}
\label{alg:procedure}

Run $j$ starts with all arms active, $k=0$, and both counters zero.
It repeats the following steps until it aborts or returns an arm.
\begin{enumerate}
\item Compute $s,d,u,w$ from \eqref{alg:work}. If the used weight plus
$w$ exceeds $M_j$, abort the run before taking any samples.
Otherwise compute $\alpha$ and $m$ from
\eqref{alg:alpha}--\eqref{alg:arm-samples}.

\item Compute the entire deterministic sample count for this call.
At scale entry it is
\[
 sm+m_z+N_{\mathrm{PAC}}(s,d/8,\alpha/16);
\]
at a later call at the same scale it is $sm$. If the current reserved
sample count plus this quantity exceeds $Q_j$, abort before taking any
samples for the call. Otherwise reserve this entire quantity and charge
$w$ to the weight counter.

\item At scale entry, select the reference and obtain $z$ as specified
above. At every call, obtain independent empirical means $(x_i)_{i\in S}$,
each from $m$ fresh samples, and form
\begin{equation}
 R_0=\{i\in S:x_i\ge z-d/2\}.
 \label{alg:raw-set}
\end{equation}
If $|R_0|<\ceil[\big]{s/2}$, add arms in decreasing order of their $x_i$ values
until the resulting set has size $\ceil[\big]{s/2}$. Call the resulting set
$R$; if $|R_0|\ge\ceil[\big]{s/2}$, take $R=R_0$.

\item Replace $S$ by $R$. If $|R|=1$, return its arm and stop the entire
algorithm. Otherwise, if $|R|\le\ceil[\big]{s/2}$, make the next call at the
same scale. If $|R|>\ceil[\big]{s/2}$, increase $k$ by one and make a new
scale-entry call.
\end{enumerate}
On abort, discard the run and start run $j+1$. A fixed arm order
resolves any ties; ties between independent Gaussian estimates occur
with probability zero.

The retained set always contains both $R_0$ and the empirical upper
$\ceil[\big]{s/2}$ arms. In particular, the retained set is enlarged before testing
for a singleton. A singleton output from a nontrivial input therefore
occurs only after a call with $s=2$. Same-scale continuation is an exact
halving up to rounding and decreases a size of at least three by a
factor at most $3/4$. Each run makes finitely many calls: it can
continue only finitely many times at one scale, and eventually
$2\cdot4^k>M_j$ prevents another scale from being entered.

The reservation order ensures, on every trajectory, including trajectories
with mistaken eliminations,
\begin{equation}
 \sum_{t\text{ in run }j}w_t\le M_j,\qquad
 \sum_{t\text{ in run }j}\alpha_t\le\gamma,
 \qquad T_j\le Q_j,
 \label{alg:pathwise-caps}
\end{equation}
where $T_j$ is the run's sample count. In particular all executed
calls have $\alpha\le\gamma$. The deterministic PAC budget is needed
here to reserve the complete first call before observing its data.

\subsection{Reference estimation and underestimation events}
\label{alg:local-errors}

Let $i_{[1]}$ be the original best arm, $\mu_{[1]}=\mu_{i_{[1]}}$,
and $\Delta_i=\mu_{[1]}-\mu_i$, with $\Delta_{[1]}=0$.
The one-sided Gaussian inequality
\begin{equation}
 \P\bigl(\overline X-\mu\ge t\bigr)\le e^{-mt^2/2},\qquad
 \P\bigl(\overline X-\mu\le-t\bigr)\le e^{-mt^2/2}
 \label{alg:gaussian-tail}
\end{equation}
holds for the mean of $m$ unit-variance Gaussian observations.
Conditional on a scale-entry history and on the selected reference arm,
\eqref{alg:reference-samples} implies
\[
 \P\bigl(|z-\mu_a|>d/16\bigm|\text{entry history},a\bigr)\le\beta_{j,k}.
\]
Let $\mathcal F_z$ be the event that any such estimate fails anywhere
in the stream of runs. Adaptively visited scales satisfy the same
conditional bound, so
\begin{equation}
 \P(\mathcal F_z)
 \le\sum_{j,k\ge0}\beta_{j,k}
 =\frac{\delta}{64}\biggl(\sum_{m\ge1}m^{-2}\biggr)^2
 <\frac{\delta}{16}.
 \label{alg:global-reference-error}
\end{equation}

The reference upper bound is the event $z\le\mu_a+d/16$.
Since every original arm has mean at most $\mu_{[1]}$, the reference upper bound
implies, regardless of which arms remain active,
\begin{equation}
 z-d/2\le\mu_{[1]}-7d/16.
 \label{alg:upper-threshold}
\end{equation}
Consequently, removing an active arm with $\Delta_i\le d/8$ requires
\begin{equation}
 x_i-\mu_i<-5d/16.
 \label{alg:severe-flag}
\end{equation}
We call \eqref{alg:severe-flag} an \emph{underestimation event}, whether or not
the arm is actually removed. Enlarging the retained set can only make removal harder.
Writing $\mathcal H_t$ for the information before the fresh arm
estimates in call $t$, including its reference data, gives
\begin{equation}
 \P\bigl(\text{arm }i\text{ has an underestimation event in call }t\bigm|\mathcal H_t\bigr)
 \le(\alpha_t/128)^{25}.
 \label{alg:severe-probability}
\end{equation}
Indeed, $512(5/16)^2/2=25$. This bound is about fresh Gaussian errors
and needs no conditioning on the reference upper bound. The reference upper bound is used only to deduce that eliminating a near-optimal arm requires an underestimation event. This distinction will be
essential in the correctness proof.

\section{Correctness through sets of near-optimal arms}
\label{sound:section}

The per-run bound $\sum_t\alpha_t\le\gamma$ is not by itself a
correctness guarantee for infinitely many runs. We instead charge
errors in calls with at most seven arms globally and control the remaining errors using a
fixed nested family of sets determined by the instance. Throughout this
section, put
\begin{equation}
 \eta=\gamma/128=\delta/131072,\qquad p=2\eta^{25}.
 \label{sound:parameters}
\end{equation}

\subsection{Stochastic domination for adaptive sampling}
\label{sound:domination-section}

\begin{lemma}[Stochastic domination by independent Bernoulli variables]
\label{sound:domination}
For one fresh run, let $F_i$ indicate that the empirical mean of original arm $i$ satisfies \eqref{alg:severe-flag} at some executed call. On an extension of the probability
space there are independent Bernoulli variables $(Y_i)_{i=1}^n$, each
with parameter at most $p$, such that $F_i\le Y_i$ simultaneously.
\end{lemma}
\begin{proof}
Give each original arm an independent rate-one Poisson process $\mathsf N_i$,
and run these processes on one common clock. Suppose the current clock
is $\theta$ immediately before call $t$'s fresh arm estimates. The exact
conditional probability of \eqref{alg:severe-flag} is
\[
 p_t=\Phi\bigl(-5d_t\sqrt{m_t}/16\bigr),
\]
where $\Phi$ is the standard Gaussian distribution function. It is
the same for all active arms because their centered empirical means
have the same variance. The parameters $d_t,m_t$ are predictable.
Advance the common clock by $\lambda_t=-\log(1-p_t)$. For each
active arm $i$, use the event
\[
 \mathsf N_i(\theta+\lambda_t)-\mathsf N_i(\theta)\ge1
\]
as its underestimation indicator. Using additional independent random
variables, generate $x_i-\mu_i$ from its Gaussian law conditional on
the indicated event or its complement. Its unconditional law is exactly
$\mathcal N(0,1/m_t)$, and the empirical means are independent across active
arms conditional on the past. If desired, the individual observations
can be generated from their conditional Gaussian distribution given
the mean; the algorithm uses only the mean. Inactive arms receive dummy
Poisson increments, which are never revealed to the algorithm. All PAC
and reference observations use independent streams.

This construction can be iterated adaptively. To be explicit, enlarge
the natural joint Poisson filtration by the independent auxiliary
randomness. Each current clock endpoint is a stopping time for this
filtration. Adding a nonnegative duration measurable at that endpoint
again gives a stopping time. The strong Markov property of the joint
Poisson process therefore gives fresh independent increments at the
next call. Conditional Gaussian generation reveals no future Poisson
increments. This proves that the construction realizes precisely the
law of the algorithm's observations, despite adaptive active sets,
scales, and sample counts.
The conditional Gaussian-law assertions refer to the algorithm's
observed-history filtration, in which unused auxiliary variables remain
fresh. The enlarged Poisson filtration is used only to justify stopping
times and independent Poisson increments; it may reveal the independent
auxiliary reservoir at time zero without changing those properties.

By \eqref{alg:severe-probability},
$p_t\le(\alpha_t/128)^{25}<1/2$. Hence every sample path satisfies
\begin{align}
 \sum_t\lambda_t
 &\le2\sum_t p_t
 \le2\sum_t(\alpha_t/128)^{25}
 \le2\biggl(\frac{\sum_t\alpha_t}{128}\biggr)^{25}
 \le2\eta^{25}=p.
 \label{sound:clock-cap}
\end{align}
Set $Y_i=\1\bigl\{\mathsf N_i(p)\ge1\bigr\}$. These variables are independent, each
has probability $1-e^{-p}\le p$, and they dominate the indicators $F_i$, since the common clock never
exceeds $p$. This is an unconditional coupling;
in particular, no future reference upper-bound event is conditioned upon.
\end{proof}

\subsection{Errors in calls with at most seven arms}
\label{sound:small-section}

Let $\mathcal F_{\mathrm{small}}$ denote the event that underestimation
as in \eqref{alg:severe-flag} occurs in any call with $2\le s\le7$ in the entire algorithm. The
piecewise rule and weight bound give the pathwise bound
\[
 \sum_{t:s_t\le7}\alpha_t\le\gamma(j+1)^{-2}
\]
within run $j$. Conditional union bounds over the arms and calls,
followed by $\sum_t a_t^{25}\le\bigl(\sum_t a_t\bigr)^{25}$ for $a_t\ge0$,
therefore imply
\begin{align}
 &\P\bigl(\text{an underestimation event at size at most seven in run }j
       \bigm|\text{run }j\text{ is reached}\bigr)
       \nonumber\\
 &\hspace{25mm}\le
 7\E\sum_{t:s_t\le7}(\alpha_t/128)^{25}
 \le7\eta^{25}(j+1)^{-50}.
 \label{sound:small-attempt}
\end{align}
Dropping the probability of reaching a run can only increase the
sum. Since $\sum_{m\ge1}m^{-50}\le1+1/49<2$,
\begin{equation}
 \P(\mathcal F_{\mathrm{small}})\le14\eta^{25}.
 \label{sound:small-global}
\end{equation}
We charge this event once and exclude it before summing over the sets
defined below. This avoids a factor depending on the number of distinct sets.

\subsection{A deterministic family of sets of near-optimal arms}
\label{sound:cores-section}

Fix the original instance and form the sets
\begin{equation}
 B_k=\bigl\{i:\Delta_i<2^{-k}/8\bigr\},\qquad k\ge0.
 \label{sound:core-family}
\end{equation}
Let $\mathcal B$ be this family with repeated sets removed. It is nested
and consists of nonempty subsets containing $i_{[1]}$. There is at most
one member of $\mathcal B$ of any given cardinality. In particular,
long intervals of empty gap scales do not enlarge this family.

For fixed $B\in\mathcal B$, write
\begin{equation}
 r=|B|,\qquad h=\sum_{i\notin B}\Delta_i^{-2},\qquad
 R_r=\max\{1,r-1\},\qquad A_r=rp^{R_r}.
 \label{sound:core-parameters}
\end{equation}
Associate an output with $B$ when its final call is at a scale $k$
with $B_k=B$. We always fix $B$ first and analyze this output
event on the original probability space; we never condition the
distribution of observations on selection of $B$.

Consider an incorrect singleton output associated with $B$ and
references satisfying the upper bound throughout its run. Every earlier tolerance
is at least the final tolerance, so every removed arm in $B$ has an underestimation event by \eqref{alg:upper-threshold}--\eqref{alg:severe-flag}.
If $r\ge2$, a singleton leaves at most one of the $r$ arms in $B$.
If $r=1$, an incorrect output must remove $i_{[1]}$. Thus
Lemma~\ref{sound:domination} and a union bound over the possible arm in $B$ with no underestimation
event give
\begin{equation}
 \P\biggl(\substack{\text{incorrect output associated with }B\text{ in run }j,\\
                       \text{all reference upper bounds hold in that run}}\biggr)
 \le A_r.
 \label{sound:core-risk}
\end{equation}
For $r\ge2$ this is the bound $rp^{r-1}$; for $r=1$ it is $p$.

\subsection{Small and intermediate weight bounds}
\label{sound:small-caps-section}

Fix $B$ and a run with weight budget $M=M_j$. Suppose it has an output
associated with $B$, at tolerance $d_*$. Its final call begins with at
least two arms and is included in the weight budget. Every outside arm
has $\Delta_i\ge d_*/8$. Therefore
\begin{equation}
 D_{\mathrm{out}}:=\max_{i\notin B}\Delta_i^{-2}
 \le64d_*^{-2}\le32M,
 \label{sound:outside-max}
\end{equation}
whenever the outside set is nonempty. Both $D_{\mathrm{out}}$ and $M$
are deterministic for this fixed $B$ and run. If their inequality
fails, no output can be associated with $B$; thus use of
\eqref{sound:outside-max} below involves no conditional selection of a
random weight bound.

Consider an arm $i\notin B$ eliminated at tolerance $d_t\le8\Delta_i$.
Its complexity satisfies $\Delta_i^{-2}\le64d_t^{-2}$. Charge this
complexity to the contribution $d_t^{-2}$ of arm $i$ to the weight of
its unique elimination call. The total complexity of all such arms is
at most $64M$. If the output is an outside arm, its complexity
is at most $32M$ by \eqref{sound:outside-max}. It follows that when
\begin{equation}
 M\le h/192,
 \label{sound:small-cap-range}
\end{equation}
outside arms eliminated at tolerances $d_t>8\Delta_i$ have total
complexity at least $h-96M\ge h/2$. Each of them has an underestimation
event if its reference satisfies the upper bound.

Here is the required weighted Bernoulli estimate, including its
constants. Write $a_i=\Delta_i^{-2}\le32M$ and $v_i=a_i/(32M)\le1$
for $i\notin B$. Since $e^{\lambda v}-1\le v(e^\lambda-1)$ on
$0\le v\le1$, independent dominating variables satisfy
\begin{align*}
 \E\exp\biggl(\lambda\sum_{i\notin B}v_iY_i\biggr)
 &\le\exp\biggl(p\sum_{i\notin B}\bigl(e^{\lambda v_i}-1\bigr)\biggr)\\
 &\le\exp\Bigl(\frac{ph}{32M}(e^\lambda-1)\Bigr).
\end{align*}
Choose $e^\lambda=1/(2p)$ and apply exponential Markov at $h/(64M)$.
This gives
\begin{equation}
 \P\biggl(\sum_{i\notin B}a_iY_i\ge h/2\biggr)
 \le \exp\Bigl[-\frac{h}{64M}\log\frac1{2p}
                    +\frac{h}{64M}\Bigr]
 =(2ep)^{h/(64M)}.
 \label{sound:weighted-tail}
\end{equation}
The dominating Bernoulli variables indexed by $B$ and its complement
are independent because these sets are fixed and disjoint. Intersecting the underestimation requirement with
the requirement on the total complexity of outside arms therefore bounds the relevant incorrect
output probability by
\begin{equation}
 A_r(2ep)^{h/(64M)}.
 \label{sound:small-cap-risk}
\end{equation}
Over the dyadic budgets satisfying \eqref{sound:small-cap-range}, the
exponents are at least three and double as budgets decrease. With
$a_0=2ep<1/4$, their sum is bounded by
$\sum_{t\ge0}a_0^{3\cdot2^t}\le a_0^3/\bigl(1-a_0^3\bigr)<1$.
Consequently all small-budget runs contribute at most $A_r$ for this
fixed $B$.

The intermediate range $h/192<M<h$ contains at most eight dyadic budgets.
Using \eqref{sound:core-risk} at each of them gives total contribution
at most $8A_r$. These conclusions also hold when a range is empty or
its lower endpoint lies below the first budget.

\subsection{A bound on eliminating near-optimal arms in one call}
\label{sound:collapse-section}

For large budgets we need to control a call that simultaneously removes
several arms in $B$. Each call uses the confidence parameter computed at its start; this
parameter is not recomputed within the call.

\begin{lemma}[Eliminating all but at most one near-optimal arm]
\label{sound:collapse}
Fix a set $B$ containing $i_{[1]}$ and write
$h=\sum_{i\notin B}\Delta_i^{-2}$. Consider an adaptive run whose total weight is at most $M$, where
$M\ge h>0$. Each call uses fresh estimates \eqref{alg:arm-samples},
follows the retention rule \eqref{alg:raw-set}, and has
$\alpha_t\le\gamma w_t/M$. A call is eligible if $s\ge8$, it has
at least one active arm in $B$, and
$\max_{i\in B}\Delta_i\le d/8$. An \emph{elimination event} occurs when a call retains at most
one arm in $B$ when at least two were active, or retains no arm in $B$
when exactly one was active. The probability of any eligible elimination event
whose reference satisfies the upper bound is at most
\begin{equation}
 3\eta^4(h/M)^3.
 \label{sound:collapse-bound}
\end{equation}
If $h=0$, an eligible elimination event is impossible.
\end{lemma}
\begin{proof}
First assume $h>0$ and expose the history of an eligible call, including
any fresh reference data. Its parameters and whether the reference upper bound holds
are now fixed, while its arm estimates remain independent. Set
\[
 q=|S\cap B|,\qquad x=w/h,\qquad \rho=\alpha/128.
\]
On histories satisfying the reference upper bound, an elimination event requires one underestimation event if $q=1$,
and at least $q-1$ underestimation events for arms in $B$ if $q\ge2$. Thus the conditional
probability of an elimination event together with the reference upper
bound is at most
\begin{equation}
 \begin{cases}
 \rho^{25},&q=1,\\
 q\rho^{25(q-1)},&q\ge2
 \end{cases}
 \le2\rho^{25}.
 \label{sound:collapse-severe}
\end{equation}
The last inequality follows from $\rho\le\eta\le1/1280$; more
generally $qa^{q-1}\le2a$ for $q\ge2$ and $0<a\le1/2$.
When the reference upper bound fails, the intersected event is empty.

For a separate bound based on the ordering of empirical means, suppose $x\ge2$ and define
$g^2=s/(8h)$. At most $hg^2=s/8$ outside arms in the entire original
instance have gap below $g$, since each contributes more than $g^{-2}$
to $h$. Also $g^2/d^2=x/8\ge1/4$, so every arm in $B$ has mean at least
$\mu_{[1]}-d/8\ge\mu_{[1]}-g/4$. The Gaussian tail bound yields
\begin{align}
 \P\bigl(x_i<\mu_{[1]}-g/2\bigm|\mathcal H_t\bigr)
 &\le e^{-m g^2/32}\le\rho^{2x},&&i\in S\cap B,
 \label{sound:core-low}\\
 \P\bigl(x_i\ge\mu_{[1]}-g/2\bigm|\mathcal H_t\bigr)
 &\le e^{-m g^2/8}\le\rho^{8x},&&i\in S\setminus B,
                         \ \Delta_i\ge g.
 \label{sound:outside-high}
\end{align}
The exponents follow by substituting $m\ge512d^{-2}\log(1/\rho)$
and $g^2=s/(8h)$.

When $q\ge2$, first consider the case in which at most one estimate of an arm in $B$ is at least $\mu_{[1]}-g/2$. At least $q-1$ independent
estimates of arms in $B$ then satisfy \eqref{sound:core-low}. When $q=1$, define
this case instead to mean that its sole estimate is low. In either
case its conditional probability is at most $2\rho^{2x}$, by the same
elementary inequality used in \eqref{sound:collapse-severe}.

In the complementary case, the elimination event requires omitting an arm
in $B$ whose estimate is
at least $\mu_{[1]}-g/2$: there are at least two such estimates of arms in $B$
when $q\ge2$, or the sole estimate of an arm in $B$ is high when $q=1$. Every
member of the empirical upper half outranks the omitted arm. At most
one member of that upper half is an arm in $B$, and at most $s/8$ outside
arms have true gap below $g$. Therefore at least
\[
 \ceil[\big]{s/2}-1-s/8\ge s/4
\]
members of the upper half satisfy \eqref{sound:outside-high}. This
count is only stronger when $q=1$, since no arm in $B$ is retained.
Let $v=\ceil[\big]{s/4}$. A union bound over $v$ of the at most $s$ far
arms, followed by $\binom{s}{v}\le(es/v)^v$, gives probability at most
\[
 (4e\rho^{8x})^v\le(4e\rho^{8x})^{s/4}\le\rho^{2x}.
\]
Indeed $4e\rho^{8x}\le\rho^{4x}<1$ for $\rho\le1/1280$ and
$x\ge2$, and $s\ge8$. Combining the two cases gives the bound based on empirical means
$3\rho^{2x}$. Combining it with \eqref{sound:collapse-severe}, using
the latter alone for $x<2$, proves the conditional bound
\begin{equation}
 \P\biggl(\substack{\text{eligible elimination event,}\\
                       \text{reference upper bound holds}}\,\biggm|\,\mathcal H_t\biggr)
 \le3\rho^{\max(25,2x)}.
 \label{sound:conditional-collapse}
\end{equation}

Put $U=M/h\ge1$. A permitted call has $0<x\le U$ and
$\rho\le\eta x/U$. We claim that
\begin{equation}
 (\eta x/U)^{\max(25,2x)}\le(\eta/U)^4x.
 \label{sound:normalization}
\end{equation}
For $x\le1$, use $\eta/U\le1$ and $x^{25}\le x$. For
$1\le x\le25/2$, divide the left side by the right side and use
$U\ge x$ to obtain a ratio at most $\eta^{21}x^3<1$. For
$x\ge25/2$, the ratio is at most $\eta^{2x-4}x^3$. Its logarithmic
derivative $2\log\eta+3/x$ is negative, and its value at $25/2$ is
less than one. This proves \eqref{sound:normalization} in all cases.

Eligibility is determined before the call's fresh arm estimates. A
union bound using predictable conditional probabilities, followed by
the pathwise weight bound, now gives
\[
 \P\biggl(\substack{\text{some eligible elimination event occurs}\\
                       \text{with its reference upper bound holding}}\biggr)
 \le3(\eta/U)^4\E\sum_{t\text{ eligible}}\frac{w_t}{h}
 \le3(\eta/U)^4 U=3\eta^4(h/M)^3.
\]
No future reference upper-bound event has been conditioned upon. Finally, when
$h=0$, the assumption of a unique best arm and positivity of every outside
arm complexity imply that no outside arms exist. All active arms belong to
$B$, and retaining the empirical upper half of a set of size $s\ge8$ keeps at least four
arms in $B$. The elimination event is therefore impossible in that case.
\end{proof}

\subsection{Large weight bounds and the final error bound}
\label{sound:completion-section}

Exclude $\mathcal F_z\cup\mathcal F_{\mathrm{small}}$, and consider
an incorrect output associated with a fixed $B$. If $r\ge2$, take
the first call at which the number of active arms in $B$ becomes at most
one. If $r=1$, take the call removing $i_{[1]}$. This is an elimination event
of the type in Lemma~\ref{sound:collapse}. Every arm in $B$ has gap less than one eighth of the current tolerance,
since $B$ was defined at the final, smallest tolerance. The call must have $s\ge8$: otherwise a
removed arm in $B$ would create an underestimation event in
$\mathcal F_{\mathrm{small}}$.

If $h=0$, such an output is impossible. If $h>0$ and $M\ge h$,
\eqref{sound:core-risk}, Lemma~\ref{sound:collapse}, and
$\min(u,v)\le\sqrt{uv}$ bound its probability by
\begin{equation}
 \sqrt{3A_r}\,\eta^2(h/M)^{3/2}.
 \label{sound:large-cap-risk}
\end{equation}
This step requires no independence between the underestimation events
for arms in $B$ and the elimination event. Over dyadic budgets at least $h$,
\[
 \sum_{j:M_j\ge h}\sqrt{3A_r}\,\eta^2(h/M_j)^{3/2}
 \le\frac{\sqrt3}{1-2^{-3/2}}\eta^2\sqrt{A_r}
 <3\eta^2\sqrt{A_r}.
\]

For completeness, the sums over cardinalities of the near-optimal sets have absolute bounds
\begin{equation}
 \sum_{r\ge1}A_r\le4p,\qquad
 \sum_{r\ge1}\sqrt{A_r}\le5\sqrt p.
 \label{sound:core-series}
\end{equation}
The first series equals $p+(1-p)^{-2}-1$ and is at most $4p$ for
$p\le1/16$. For the second, put $t=\sqrt p\le1/4$ and use
$\sqrt r\le r$ to bound it by
$t+(1-t)^{-2}-1\le5t$. Our value of $p$ satisfies both restrictions.
Since $\mathcal B$ has at most one set of any given size, the small,
intermediate, and large ranges together contribute at most
\[
 \sum_{r\ge1}\bigl(9A_r+3\eta^2\sqrt{A_r}\bigr)
 \le36p+15\eta^2\sqrt p
\]
to the probability of an incorrect output outside the two excluded
events. Each single-run bound remains valid conditional on reaching
that run, since it starts with fresh data. Dropping reach
probabilities justifies all sums over runs above.

Adding \eqref{alg:global-reference-error} and
\eqref{sound:small-global}, the total probability of any incorrect
output is at most
\begin{align}
 \frac{\delta}{16}+14\eta^{25}+36p+15\eta^2\sqrt p
 &=\frac{\delta}{16}+86\eta^{25}+15\sqrt2\,\eta^{29/2}
 \nonumber\\
 &<\delta.
 \label{sound:final-error}
\end{align}
For the last inequality, $0<\eta<1$ allows both positive powers to be
replaced by $\eta$, and
$(86+15\sqrt2)\eta<108\delta/131072$.
This bounds the probability of an incorrect output by $\delta$ on every
admissible instance. Combining \eqref{sound:final-error} with the
almost-sure termination proved in Appendix~\ref{cost:section}
establishes $\delta$-correctness in the sense of Section~\ref{sec:model}:
$\P_I\bigl(\tau<\infty,\widehat{\imath}=i_{[1]}\bigr)\ge1-\delta$.

\section{Expected sample complexity without instance parameters}
\label{cost:section}

We prove the expected sample complexity bound of Theorem~\ref{thm:universal}.
The analysis uses sample paths in a good event to show that sufficiently large
runs succeed with constant probability. Deterministic sample budgets
then control every other path, including paths on which the best arm
is eliminated.

\subsection{A good event in a fresh run}
\label{cost:good-event-section}

First consider a variant of run $j$ that enforces only its weight
budget, omitting its sample budget. All confidence and sampling rules remain
unchanged. If the best arm is active at a scale entry, accurate PAC
selection and an accurate reference estimate imply
\begin{equation}
 \mu_{[1]}-3d/16\le z\le\mu_{[1]}+d/16.
 \label{cost:reference-interval}
\end{equation}
As long as the original best remains active, this interval is suitable
for every subsequent call at that scale, even if the reference arm is
removed. Let $N(d)=\bigl|\{i:\Delta_i<d\}\bigr|$, including the original best, so $N(d_k)=N_k$ in the notation of Appendix~\ref{sec:prelim}.

We give explicit failure bounds for the modified elimination procedure of
Lemma~\ref{lem:elimination}. Given
\eqref{cost:reference-interval}, an empirical mean of the best within
$d/16$ of its true mean protects it. By
\eqref{alg:gaussian-tail} and \eqref{alg:arm-samples}, the probability
of failure of this sufficient condition is at most $\alpha/64$.
Any arm with $\Delta_i\ge d$ can enter the threshold set only by an upward
error of at least $5d/16$, whose probability is at most
$(\alpha/128)^{25}\le\alpha/64$. If $s>4N(d)$ and the threshold set has
more than $\ceil[\big]{s/2}$ members, more than $s/4$ far arms have such
errors. Markov's inequality bounds this event by $\alpha/16$.
Otherwise the retention rule returns exactly $\ceil[\big]{s/2}$ arms, so the call
halves. When $s\le4$ and every suboptimal active gap is at least $d$,
the sufficient best-arm retention condition and the absence of all far-arm
upward errors give the threshold set exactly $\{i_{[1]}\}$. A union bound makes
the failure probability at most $\alpha/16$. Thus the modified retained
set contains the best arm and halves, and at most two such calls finish from size at most four.

Charge a scale's PAC failure once, at its entry call, at most
$\alpha/16$. Charge best-arm retention and required progress failures at
every call as above. Their combined conditional probability is less
than $\alpha$. Define $\mathcal G_j$ by the absence of these failures
and the accuracy of every visited reference estimate. The definition
can be made along the successive prefixes until the first failure;
the conditional bounds apply while the best is still active. The
pathwise inequality $\sum_t\alpha_t\le\gamma$ and the reference
union bound give
\begin{equation}
 \P(\mathcal G_j)\ge1-q_0,\qquad
 q_0=\gamma+\delta/16<1/2,
 \label{cost:good-probability}
\end{equation}
uniformly for every fresh run. Here $\delta/16$ is the global
bound \eqref{alg:global-reference-error} on all reference-estimation
failures in the whole stream, which dominates the failure probability
of the reference estimates of a single run. This is a per-run
sample complexity statement; we do not sum it over runs to prove correctness.

\subsection{Deterministic bounds on the weights of calls}
\label{cost:envelope-section}

For analysis, let $n_t$ count suboptimal gaps in
$[2^{-t},2^{-t+1})$, for $t\ge0$, and let $K$ be the largest index of a nonempty group. Define
\begin{equation}
 W_k=4^k\biggl(1+\sum_{t\ge k}n_t\biggr),\qquad
 W=\sum_{k=0}^K W_k,\qquad q_k=W_k/W.
 \label{cost:smoothed-work}
\end{equation}
Lemma~\ref{lem:smoothing} states
\begin{equation}
 H\le W<\frac{32}{3}H,\qquad
 \Ent(q)\le\frac{32}{3}\Ent(I)+C_0,
 \qquad 4^K\le4D,
 \label{cost:smoothing-bounds}
\end{equation}
where $\Ent(q)$ is Shannon entropy and $C_0$ is universal.

On $\mathcal G_j$, a scale can stall only if its input size is at most
$4N(d)$. At scale zero the entry size is $n$. At scale $k\ge1$, the
preceding scale could have ended only by such a stall, so the entry
size is at most
\[
 4N(2^{-(k-1)})
 =4\biggl(1+\sum_{t\ge k}n_t\biggr).
\]
An accepted same-scale continuation decreases its input size by at
least a factor $3/4$. If $r=0,1,\ldots$ indexes the successive calls
at scale $k$, this gives
\begin{equation}
 w_{k,r}\le v_{k,r}:=4W_k(3/4)^r,
 \qquad
 V:=\sum_{k=0}^K\sum_{r\ge0}v_{k,r}=16W<171H.
 \label{cost:work-envelope}
\end{equation}
At scale $K$, all suboptimal gaps are at least $2^{-K}$. More than
four arms therefore halve, and at most four arms finish in at most two
further calls. Thus a sample path in $\mathcal G_j$ finishes by scale $K$,
provided the weight bound permits those calls.

This last proviso does not assume the conclusion: the same upper bounds apply to every executed prefix in
$\mathcal G_j$ and its prospective next call. If $h_j\ge H$, their cumulative prospective weight is below
$171H<M_j=1024h_j$. No such prefix can abort at the weight bound.
This proves both completion by scale $K$ and the bounds on call weights for the
process constrained only by the weight bound.

\subsection{Sample complexity on the good event}
\label{cost:declared-section}

Write $h=h_j\ge H$. By Lemma~\ref{lem:pac}, the declared arm and PAC
cost of a call is at most $Cw\log(128/\alpha)$ for a universal
constant $C$; assigning this cost at every call also bounds the PAC
cost, which occurs only at scale entry. Temporarily omit the additional
small-size term $2\log(j+1)$ in $\log(128/\alpha)$. The resulting
baseline cost is bounded by a constant times
\[
 \sum_{k,r}w_{k,r}\log\frac{128M_j}{\gamma w_{k,r}}.
\]
The function $f(x)=x\log\bigl(128M_j/(\gamma x)\bigr)$ is increasing on
$(0,M_j]$. Each upper bound $v_{k,r}$ lies below
$V<171H<M_j$, so \eqref{cost:work-envelope} permits termwise
replacement of $w_{k,r}$ by $v_{k,r}$. Let
$\pi_r=\frac14(3/4)^r$; normalizing these upper bounds gives
$v_{k,r}/V=q_k\pi_r$. Hence
\begin{align}
 \sum_{k=0}^K\sum_{r\ge0}f(v_{k,r})
 &=V\Bigl[\log\frac{128M_j}{\gamma V}
                +\Ent(q)+\Ent(\pi)\Bigr]
 \nonumber\\
 &\le C H\bigl[L+\Ent(I)+1+\log(h/H)\bigr].
 \label{cost:baseline}
\end{align}
Here $16H\le V<171H$, $M_j=1024h$,
$\gamma=\delta/1024$, \eqref{cost:smoothing-bounds} holds, and
the geometric distribution $\pi$ has bounded entropy.

Only calls with $s\le7$ incur the additional sample cost due to the
factor $(j+1)^{-2}$ in the confidence allocation. Their extra
cost is at most $2Cs4^k\log(j+1)$ per call. At any fixed scale, once
the size is at most seven, successive call sizes follow a sequence
bounded by $7,4,2$: a continuation halves rounded up, and a stall ends
the scale. This includes the last possibly stalled call. Consequently
\begin{equation}
 \sum_{\substack{t\text{ at scale }k\\s_t\le7}}s_t\le13,
 \qquad
 \text{additional sample cost}
 \le C D\log(j+1).
 \label{cost:small-surcharge}
\end{equation}
The PAC procedure contributes this additional cost only if the entry call
has at most seven arms, so it is already covered by the same bound.
Calls of size at least eight incur no such additional cost.

There is one reference estimate per visited scale. Its declared cost
from \eqref{alg:reference-samples} is at most
\begin{align}
 C\sum_{k=0}^K4^k\bigl[L+1+\log(j+1)+\log(k+1)\bigr]
 &\le C D\bigl[L+\log(j+1)+\ell(D)\bigr],
 \label{cost:reference-cost}
\end{align}
where $\ell(D)=\log\bigl(e+\log(e+D)\bigr)$. Indeed,
$\sum_{k=0}^K4^k\le(4/3)4^K\le16D/3$, and
$K\le\log_4(4D)$ makes $\log(K+1)\le C\ell(D)$.
Integer ceilings in all sample counts are absorbed by these estimates:
$4^k\ge1$ and all confidence logarithms exceed one.

Combining \eqref{cost:baseline}--\eqref{cost:reference-cost}, on $\mathcal G_j$, every run constrained only
by the weight bound with $h_j\ge H$ has total declared sample cost at most
\begin{equation}
 B_j^{\mathrm{good}}=C_{\mathrm{good}}\Bigl\{
 H\bigl[L+\Ent(I)+1+\log(h_j/H)\bigr]
 +D\bigl[L+\log(j+1)+\ell(D)\bigr]\Bigr\},
 \label{cost:good-cost}
\end{equation}
for a fixed universal $C_{\mathrm{good}}$. In particular, this bounds the sum of
the complete prospective call reservations, not only samples observed
before a possible sample-budget rejection.

\subsection{Choosing the fixed budget constant and summing all trajectories}
\label{cost:expectation-section}

Define, for analysis only,
\begin{equation}
 \Psi=H\bigl[L+\Ent(I)+1\bigr]+D\ell(D),\qquad
 j_* =\ceil[\big]{\log_2(\Psi/L)}.
 \label{cost:threshold}
\end{equation}
Because the means lie in $[0,1]$ and $n\ge2$, we have $H\ge D\ge1$.
Thus $j_*\ge0$, $h_{j_*}\in[\Psi/L,2\Psi/L]$, and
$h_j\ge H$ for all $j\ge j_*$. For such $j$, set $h=h_j$. Then
\[
 \Psi\le hL,\qquad H\log(h/H)\le h.
\]
We also have the uniform estimate
\begin{equation}
 \log(j+1)\le C+\ell(D)+\log(h/D).
 \label{cost:guess-log}
\end{equation}
To see this, put $y=\log(h/D)\ge0$. Since
$j=\log h/\log2$, the quantity $j+1$ is at most
$C(1+\log D+y)$. The inequality
$1+\log D+y\le(1+\log D)(1+y)$ and
$\log(1+y)\le y$ yield \eqref{cost:guess-log}. Finally,
$D\log(h/D)\le h$. Substituting these estimates into
\eqref{cost:good-cost}, using $D\le H$ and $L>\log10$, gives
\begin{equation}
 B_j^{\mathrm{good}}\le C_*h_jL\qquad(j\ge j_*),
 \label{cost:sample-cap-fits}
\end{equation}
for an absolute constant $C_*$ determined only by $C_{\mathrm{PAC}}$,
$C_{\mathrm{elim}}$, $C_{\mathrm{good}}$, and the ratio $\gamma/\delta=1/1024$; it does
not depend on $n$, $\delta$, or the instance. Fix once and for all
$C_Q:=C_*$; any larger absolute constant also works. This choice is
noncircular: neither the analysis with only the weight bound enforced nor $C_*$ depends on
$C_Q$.

Couple the run enforcing both budgets to the run enforcing only the
weight budget, using the same samples. On $\mathcal G_j$ with $j\ge j_*$, the weight bound
permits completion, and its total declared sample cost fits $Q_j$ by
\eqref{cost:sample-cap-fits}. Every complete-call reservation therefore
passes. The event that run $j$ is reached is determined by the
observations of runs $j'<j$, while run $j$ uses fresh
observations independent of them; hence,
conditional on being reached, run $j$ has the law of a fresh run
and the bound \eqref{cost:good-probability} applies conditionally.
Conditional on being reached, each such run thus returns
the correct arm with probability at least $1-q_0$. Reaching the next
run requires failure to return, so induction gives
\begin{equation}
 \P\bigl(\text{run }j_*+t\text{ is reached}\bigr)\le q_0^t,
 \qquad t\ge0.
 \label{cost:reach}
\end{equation}
Returning an incorrect arm only decreases the probability of reaching
subsequent runs, so it does not invalidate this bound.

Let $\tau$ be the total sample count, allowing the value $+\infty$.
Tonelli's theorem and the deterministic budget $T_j\le Q_j$ give
\begin{align}
 \E_I\tau
 &\le\sum_{j<j_*}Q_j+
       \sum_{t\ge0}Q_{j_*+t}q_0^t
 \nonumber\\
 &\le C Q_{j_*}\biggl(1+\sum_{t\ge0}(2q_0)^t\biggr)
 =O(Q_{j_*})=O(\Psi).
 \label{cost:expected-total}
\end{align}
Here $2q_0<1$, the earlier budgets form a geometric sum, and their ceilings
change only the absolute constants. Each run is finite, and
\eqref{cost:reach} also shows that the probability of infinitely many
runs is zero. Thus the algorithm stops almost surely and has finite
expected number of samples on every instance with a unique best arm.

Finally, $L>\log10$ absorbs the additive $H$ in $\Psi$, proving
\[
 \E_I\tau\le C\Bigl\{H\bigl[\log(1/\delta)+\Ent(I)\bigr]
                         +D\ell(D)\Bigr\}.
\]
Together with \eqref{sound:final-error}, this completes the proof of
Theorem~\ref{thm:universal}. All quantities defining $B$, $H$, $D$,
$K$, $W$, $\Psi$, and $j_*$ are used exclusively in the proof; the
algorithm requires none of these quantities as input.

\section{Consequences and normalization}
\label{app:section}

The normalization in our model allows a direct interpretation in
physical measurement units. Suppose an observation from arm
$i$ has law $\mathcal N\bigl(\nu_i,\sigma^2\bigr)$, where $\sigma>0$ is known and
common to all arms, observations are independent, and
$\nu_i\in[a,a+\sigma]$ for a known $a$. Transform each observation
by $Y=(X-a)/\sigma$. This preserves the best arm and
maps the experiment exactly to the model of Section~\ref{sec:model}.
The transformation preserves the order $i_{[r]}$ of the arm labels. Write $\nu_{[r]}=\nu_{i_{[r]}}$, $g_i=\nu_{[1]}-\nu_i$, and $g_{[r]}=g_{i_{[r]}}$, and let $I_\sigma=\bigl((\nu_i-a)/\sigma\bigr)_{i=1}^n$ be the normalized instance. With $H_\sigma=H(I_\sigma)$ and $D_\sigma=D(I_\sigma)$, Theorem~\ref{thm:universal} gives
\begin{equation}
 \begin{aligned}
 H_\sigma=\sum_{i\ne i_{[1]}}\frac{\sigma^2}{g_i^2},
 &\qquad D_\sigma=\frac{\sigma^2}{g_{[2]}^2},\\
 \E\tau&=O\!\Bigl(
 H_\sigma\bigl[\log(1/\delta)+\Ent(I_\sigma)\bigr]
 +D_\sigma\ell(D_\sigma)\Bigr),
 \end{aligned}
 \label{app:normalization}
\end{equation}
Here $g_{[2]}$ is the smallest positive unnormalized gap, and $\Ent(I_\sigma)$
is computed from the normalized gaps $g_i/\sigma$.
Each normalized observation uses exactly one original observation, so
the complexity depends on squared noise-to-gap ratios rather than on
the origin or units of the reported scores.

The algorithm receives no estimate of these ratios or of their
distribution across scales, and its guarantee holds for every labeling
of the same arms. An independent initial random permutation
makes an implementation explicitly invariant to labels while preserving
correctness and the sample bound. The order-oblivious instance-wise lower bound therefore
measures difficulty intrinsic to the mean vector, without
rewarding a favorable input order.

In the normalized model, Theorems~\ref{thm:entropy} and~\ref{thm:universal}
give a criterion for matching the order-oblivious instance-wise lower bound:
\begin{equation}
 \frac{T_A(I)}{\cL(I,\delta)}
 =O\!\biggl(1+\frac{\ell(D)}{(H/D)\bigl[L+\Ent(I)\bigr]}\biggr).
 \label{app:benchmark-ratio}
\end{equation}
The factor $H/D=\sum_{i\ne i_{[1]}}(\Delta_{[2]}/\Delta_i)^2$
is the instance complexity divided by the complexity of an arm with
the smallest positive gap.
Whenever $(H/D)\bigl[L+\Ent(I)\bigr]\ge\ell(D)$, the additional term is
absorbed into the instance-wise lower bound. If all $n-1$ suboptimal gaps are equal,
then $H=(n-1)D$ and $\Ent(I)=0$, so the condition becomes
$(n-1)L\ge\ell(D)$: unknown-gap adaptation is absorbed by the ordinary
confidence cost even though the common gap is not known to the procedure.

Noisy model selection and hyperparameter comparison are covered only
when their score measurements meet the same assumptions: reused
validation data can create dependence between observations, losses need
not be Gaussian, and configurations can differ in noise variance or
measurement cost. The present theorems do not justify substituting such
scores into the Gaussian procedure without an observation model and
corresponding concentration arguments.

\section{Internal randomness and the scope of the algorithm model}
\label{bridge:section}

The algorithms of Section~\ref{sec:model} draw an internal random seed
$Z=(Z_0,Z_1,\ldots)$ with law $\Gamma=\mathcal N(0,1)^{\otimes\N}$, independent
of all rewards, and apply a measurable rule to the seed and the finite
observed history. This appendix records the scope of that model. The
instance-wise lower bound~\eqref{eq:benchmark} is unchanged when the seed may be drawn
from an arbitrary standard Borel probability space
(Proposition~\ref{bridge:benchmark}); the algorithms of
Appendices~\ref{pointwise:section} and~\ref{alg:section} use no internal randomness (Proposition~\ref{bridge:witnesses}); algorithms with an
explicit abandonment action and algorithms specified by history-dependent
probability kernels are covered by the lower bound
(Propositions~\ref{bridge:completion} and~\ref{bridge:kernel}).
Appendix~\ref{bridge:convention} compares the truncated iterated
logarithm $\ell$ of~\eqref{eq:notation-main} with the untruncated
expression of \citet{chen2016open}. These statements are also covered
by the formal development described in Appendix~\ref{verif:section}.

\subsection{Algorithms with a standard Borel internal random seed}
\label{bridge:seed}

Fix an arm count $n\ge2$. Write $\mathcal O_t=\bigl([n]\times\R\bigr)^t$ for the
space of observation histories of length $t$ and
\[
 \mathcal A_n=\bigl\{\mathsf{sample}(i),\mathsf{return}(i):i\in[n]\bigr\}
\]
for the finite decision space. Let $(S,\mu)$ be a standard Borel space
with a probability measure. A \emph{randomized algorithm} on $(S,\mu)$
is a family of measurable maps
\[
 a_t:S\times\mathcal O_t\to\mathcal A_n,\qquad t\ge0.
\]
The seed is drawn once from $\mu$, independently of all rewards; the
algorithm receives neither the means nor the best label. A family of
algorithms may choose $S$, $\mu$, and the decision rules separately for
each arm count and each confidence input. The model of
Section~\ref{sec:model} is the case $(S,\mu)=(\R^\N,\Gamma)$.

For an instance $I$, let $Q_I$ be the law of the reward table
$r=(r_{t,i})_{t\ge0,i\in[n]}$ of independent $\mathcal N(\mu_i,1)$ variables.
Execution uses the product law $\mu\otimes Q_I$. The state at time $0$
is the empty history. At an active state $h\in\mathcal O_t$ with seed
$s$, the decision $a_t(s,h)=\mathsf{sample}(i)$ moves to the history
$\bigl(h,(i,r_{t,i})\bigr)\in\mathcal O_{t+1}$, and $a_t(s,h)=\mathsf{return}(i)$
moves to the absorbing state $\mathsf{return}(i)$. Write
$\mathrm{run}_a(s,r,t)$ for the state at time $t$; it is a measurable
function of $(s,r)$. Since row $t$ of the table is independent of the
seed and of the earlier rows, the sampled entry is a fresh $\mathcal N(\mu_i,1)$
observation conditional on the past, as in Section~\ref{sec:model}.
Success is the event that $\mathrm{run}_a(s,r,t)=\mathsf{return}(i_{[1]})$
for some finite $t$. The sample count is the number of sampling
decisions,
\[
 \tau_a(s,r)=\sum_{t\ge0}\1\bigl\{\mathrm{run}_a(s,r,t+1)\in\mathcal O_{t+1}\bigr\}
 \in[0,\infty],
\]
which is the first return time of Section~\ref{sec:model}, and
$T_a(I)=\int\tau_a\,d(\mu\otimes Q_I)\in[0,\infty]$ is its
unconditional expectation. Neither almost-sure termination nor finite expectation
is required separately. These definitions do not refer to the
model with Gaussian random seeds or to any complexity bound.

\subsection{Representing internal randomness by a Gaussian seed}
\label{bridge:realization}

\begin{lemma}[Randomization lemma]\label{bridge:law}
For every standard Borel probability space $(S,\mu)$ there is a
measurable map $f_\mu:\R^\N\to S$ with $(f_\mu)_\#\Gamma=\mu$. The map
depends on $\mu$ only, not on the instance.
\end{lemma}
\begin{proof}
Let $\Phi$ be the distribution function of $\mathcal N(0,1)$ and $\mathsf U$ the
uniform law on $[0,1]$. Since $\Phi$ is continuous and strictly
increasing, $\Phi(z_0)$ has law $\mathsf U$ under $\Gamma$; it suffices
to compare the measures of the lower intervals $[0,u]$, including the
endpoints $u=0$ and $u=1$. By the measurable representation theorem for
probability measures on standard Borel spaces, there is a measurable
$g_\mu:[0,1]\to S$ with $(g_\mu)_\#\mathsf U=\mu$. Put
$f_\mu(z)=g_\mu\bigl(\Phi(z_0)\bigr)$. The space $S$ is nonempty because it
carries a probability measure, and no atomlessness of $\mu$ is needed.
\end{proof}

The lemma concerns measurable maps only; it does not assert that
$f_\mu$ is computable from finitely many random bits or within a
prescribed running time.

\begin{lemma}[Equality of sample paths under seed simulation]\label{bridge:transport}
Let $a$ be a randomized algorithm on $(S,\mu)$, let $f_\mu$ be as in
Lemma~\ref{bridge:law}, and define the algorithm with a Gaussian random seed
$b_t(z,h)=a_t\bigl(f_\mu(z),h\bigr)$. For every reward table $r$ and every $t$,
\begin{equation}\label{bridge:pathwise}
 \mathrm{run}_b(z,r,t)=\mathrm{run}_a\bigl(f_\mu(z),r,t\bigr).
\end{equation}
Consequently, on every instance $I$, $b$ and $a$ have the same law of
the complete state sequence $\bigl(\mathrm{run}(t)\bigr)_{t\ge0}$, the same
success probability, and the same expected number of samples in
$[0,\infty]$; and $b$ terminates almost surely if and only if $a$ does.
\end{lemma}
\begin{proof}
Each $b_t$ is measurable as a composition. Identity~\eqref{bridge:pathwise}
follows by induction on $t$: both executions start from the empty
history, a returned label is absorbing, and at an active history the
two decisions agree and a sampling decision reads the same entry of $r$.
One coupling therefore works at every time; it preserves every requested
arm and the full sample count, including on erroneous and nonreturning
paths.

The map $\psi(z,r)=\bigl(f_\mu(z),r\bigr)$ takes $\Gamma\otimes Q_I$ to
$\mu\otimes Q_I$. Pushing~\eqref{bridge:pathwise} forward gives equality
of the laws of the entire state sequences, and applying $\psi$ to the
measurable success and termination events gives the two probability
statements. Finally,
\[
 \int\tau_b\,d(\Gamma\otimes Q_I)
 =\int\tau_a\circ\psi\,d(\Gamma\otimes Q_I)
 =\int\tau_a\,d(\mu\otimes Q_I)
\]
by the change-of-variables formula for nonnegative measurable
functions, which needs no integrability. The same $b$ and $\psi$ serve
for every $I$.
\end{proof}

\subsection{Equality of the instance-wise lower bounds}
\label{bridge:equality}

Let $\mathfrak A_{\mathrm{SB}}$ be the class of families
$A=(A^{(n)})_{n\ge2}$ in which $A^{(n)}$ is a randomized algorithm for
$n$ arms on a standard Borel probability space that may depend on $n$.
Such a family is $\delta$-correct if its success probability is
at least $1-\delta$ on every admissible instance at every arm count.
Define
\begin{equation}\label{bridge:benchmark-sb}
 \cL_{\mathrm{SB}}(I,\delta)
 =\inf_{A\in\mathfrak A_{\mathrm{SB}}:\ A\text{ is }\delta\text{-correct}}
 \frac1{n!}\sum_{\pi\in\mathfrak S_n}T_A(\pi I),
\end{equation}
with all expectations and the infimum in $[0,\infty]$. The instance-wise lower bound
$\cL(I,\delta)$ of~\eqref{eq:benchmark} is the same infimum over the
algorithm families with Gaussian random seeds of Section~\ref{sec:model}.

\begin{proposition}[Equality of the instance-wise lower bounds]\label{bridge:benchmark}
For every admissible instance $I$ and every $0<\delta<1$,
$\cL_{\mathrm{SB}}(I,\delta)=\cL(I,\delta)$.
\end{proposition}
\begin{proof}
The space $\R^\N$ is Polish and $\Gamma$ is a Borel probability measure
on it, so every algorithm family with Gaussian random seeds belongs to
$\mathfrak A_{\mathrm{SB}}$ with the same success probabilities and
costs. Hence $\cL_{\mathrm{SB}}\le\cL$. Conversely, let
$A\in\mathfrak A_{\mathrm{SB}}$ be $\delta$-correct and apply
Lemma~\ref{bridge:transport} to the distribution of the internal random seed of $A^{(n)}$ for each
$n$. The resulting algorithm family with Gaussian random seeds has the same success probability
on every instance, hence is $\delta$-correct, and the same
expected sample count on every relabeling $\pi I$. Hence $\cL\le\cL_{\mathrm{SB}}$. No
minimizer, finite expectation, or termination assumption is used.
\end{proof}

Theorem~\ref{thm:entropy}, the consequence
$T_A(I)=O\bigl(\cL(I,\delta)+D\ell(D)\bigr)$ of Theorem~\ref{thm:universal}, and
Corollary~\ref{bridge:regimes} below therefore hold with their stated
constants when the instance-wise lower bound is taken over $\mathfrak A_{\mathrm{SB}}$.
This is an equality of instance-wise lower bounds, not merely the inclusion of the
algorithms with Gaussian random seeds in a larger class.

\subsection{The upper-bound algorithms need no internal randomness}
\label{bridge:deterministic}

Call an algorithm with a Gaussian random seed $b$ \emph{seed independent} if
$b_t(z,h)=b_t(z',h)$ for all $z,z'\in\R^\N$, all $h$, and all $t$.

\begin{proposition}[Deterministic algorithms]\label{bridge:witnesses}
The algorithm with a fixed target instance of Appendix~\ref{pointwise:section} and the
single algorithm of Appendix~\ref{alg:section} are both seed
independent, at every arm count and every confidence input. Each
therefore admits
measurable decision functions of the observation history alone, and
reinterpreting these functions as an algorithm with a Gaussian random seed gives back
exactly the original algorithm.
\end{proposition}
\begin{proof}
The primitives of Appendix~\ref{sec:prelim} draw deterministic numbers
of samples and compute the median-elimination output, the reference
estimate, the retained set, and the fixed tie rule from the observed
values alone. The run loops of both algorithms, their budgets and abort
rules, the scale changes, and the fallback of
Appendix~\ref{pointwise:section} are functions of the observed values
and of fixed parameters, in every parameter branch. Interleaving two
seed-independent streams sample by sample is seed independent. Following
the definitions of the two algorithms therefore proves seed independence
of the complete families. Evaluating a seed-independent measurable
algorithm at one fixed seed $z^0$ gives a measurable function of the
history alone, namely a section of a jointly measurable map, and seed
independence gives $b_t(z,h)=b_t(z^0,h)$ for every $z$.
\end{proof}

Accordingly, both algorithms belong to $\mathfrak A_{\mathrm{SB}}$ with a
one-point seed space and its Dirac law. Their correctness,
almost-sure termination, finite expected number of samples, and cost bounds
are those established in Appendices~\ref{pointwise:section}
and~\ref{alg:section}--\ref{cost:section}; this appendix does not
replace them by new procedures. The model permits arbitrary measurable
dependence on a complete internal random seed, but every external action is a
sampling or return decision. The statements above concern this standard
Borel formulation; they are not theorems about arbitrary measurable
spaces, arbitrary program syntax, or nonterminating internal computation
between two external actions.

\subsection{Explicit abandonment}
\label{bridge:abandon}

One may enlarge the decision space by an action $\mathsf{abandon}$: it
is absorbing, requests no further sample, and does not count as a
correct return. Execution, success, and $\tau$ are defined as in
Appendix~\ref{bridge:seed}, with the abandoned state as a second kind of
terminal state.

\begin{proposition}[Completion of explicit abandonment]\label{bridge:completion}
For every measurable algorithm $a$ with this additional action there is a
measurable sample-or-return algorithm $\bar a$ on the same seed space,
executed with the same seed and reward table, such that on every path
the two executions make the same sampling requests,
$T_{\bar a}(I)=T_a(I)$ on every instance, and the success probability
of $\bar a$ is at least that of $a$.
\end{proposition}
\begin{proof}
Let $\bar a_t(s,h)=\mathsf{return}(1)$ when $a_t(s,h)=\mathsf{abandon}$
and $\bar a_t(s,h)=a_t(s,h)$ otherwise. Define $\theta$ on states by
fixing every active history and every returned label and mapping the
abandoned state to $\mathsf{return}(1)$. Induction on $t$ gives
$\mathrm{run}_{\bar a}(s,r,t)=\theta\bigl(\mathrm{run}_a(s,r,t)\bigr)$: the
decisions agree at every active history, and both kinds of terminal
state are absorbing. In particular the state is active at time $t+1$
under $\bar a$ exactly when it is active under $a$, so
$\tau_{\bar a}=\tau_a$ pathwise; integrating gives the cost identity.
Every correct return of $a$ is a correct return of $\bar a$, and an
abandoned path may in addition become a correct return, so monotonicity
of measure gives the success inequality.
\end{proof}

Applied at each arm count, the completion preserves correctness
and the expected sample count on every relabeling $\pi I$, so Proposition~\ref{bridge:benchmark} extends to
families with an abandonment action. It does not preserve the complete
return law or the exact success probability. An explicit measurable
abandonment action is also distinct from unobservable nonterminating
internal computation between two external decisions; the proposition
makes no claim about the latter.

\subsection{History-dependent probability kernels}
\label{bridge:kernels}

A second formulation specifies, for each $t\ge0$, a probability kernel
$\kappa_t$ from $\mathcal O_t$ to the finite set $\mathcal A_n$, that
is, measurable functions $h\mapsto\kappa_t\bigl(h,\{d\}\bigr)$, $d\in\mathcal A_n$,
summing to one. Its state laws are defined directly, without a sampler:
$\nu_0^I$ is the Dirac mass at the empty history, and $\nu_{t+1}^I$ is
the image of $\nu_t^I$ under the transition that fixes every returned
label, moves an active $h\in\mathcal O_t$ to $\mathsf{return}(i)$ with
probability $\kappa_t\Bigl(h,\bigl\{\mathsf{return}(i)\bigr\}\Bigr)$, and with probability
$\kappa_t\Bigl(h,\bigl\{\mathsf{sample}(i)\bigr\}\Bigr)$ appends $(i,x)$ with
$x\sim \mathcal N(\mu_i,1)$. The direct success probability and unconditional
cost are
\begin{equation}\label{bridge:kernel-cost}
 p_K(I)=\sup_{t\ge0}\nu_t^I\Bigl(\bigl\{\mathsf{return}(i_{[1]})\bigr\}\Bigr),\qquad
 T_K(I)=\sum_{t\ge0}\nu_{t+1}^I(\mathcal O_{t+1}).
\end{equation}
The index $t+1$ is essential: a sample is taken at time $t$ exactly when
the state at time $t+1$ is active. The supremum over $t$ of the
$\nu_t^I$-probability of all returned labels is the total return
probability; the algorithm specified by probability kernels terminates almost surely when it equals
one.

\begin{proposition}[Implementation of randomized decision rules]\label{bridge:kernel}
Every algorithm specified by probability kernels $K=(\kappa_t)_{t\ge0}$ can be implemented by an algorithm with an internal random seed, and hence, by Lemma~\ref{bridge:transport},
by an algorithm with a Gaussian seed, chosen before the instance is known. On every admissible
instance it has success probability $p_K(I)$ and expected number of samples
$T_K(I)$, and it terminates almost surely if and only if the total
return probability of $K$ equals one. Applied at every arm count, the
implementation preserves correctness and the permutation average
of the direct costs.
\end{proposition}
\begin{proof}
Enumerate $\mathcal A_n=\{d_1,\ldots,d_{2n}\}$ and define
\[
 F_t(h,u)=d_m\quad\text{for the least $m$ with }
 u<\sum_{l\le m}\kappa_t\bigl(h,\{d_l\}\bigr),\qquad F_t(h,1)=d_{2n}.
\]
Then $F_t$ is jointly measurable and
$\mathsf U\Bigl(\bigl\{u:F_t(h,u)=d\bigr\}\Bigr)=\kappa_t\bigl(h,\{d\}\bigr)$ for every $h$ and $d$.
Take the seed space $[0,1]^\N$ with the law $\mathsf U^{\otimes\N}$ and
the decisions $a_t(u,h)=F_t(h,u_t)$. The state $\mathrm{run}_a(u,r,t)$
is a function of $u_0,\ldots,u_{t-1}$ and of the first $t$ rows of $r$,
so the pair $(u_t,r_{t,\cdot})$ is independent of it, and $u_t$ is
independent of $r_{t,\cdot}$. By independence and Tonelli's theorem, the
law of $\mathrm{run}_a(u,r,t+1)$ is the image of the law of
$\mathrm{run}_a(u,r,t)$ under the direct transition: from an active $h$
the decision $F_t(h,u_t)$ has law $\kappa_t(h,\cdot)$, and given
$\mathsf{sample}(i)$ the appended value $r_{t,i}$ has law $\mathcal N(\mu_i,1)$.
Induction gives $\mathrm{Law}_I\bigl(\mathrm{run}_a(u,r,t)\bigr)=\nu_t^I$ for
every $t$.

The events $\bigl\{\mathrm{run}_a(u,r,t)=\mathsf{return}(i_{[1]})\bigr\}$
increase in $t$ because returns are absorbing, and their union is the
success event; continuity from below gives the success probability
$p_K(I)$. The sampling indicator at time $t$ is the indicator that the
state at time $t+1$ is active, so Tonelli's theorem gives
$T_a(I)=T_K(I)$, including the value $+\infty$. The same absorbing-event
argument for all returned labels gives the termination equivalence.
Since $[0,1]^\N$ is standard Borel, Lemma~\ref{bridge:transport}
transports the implementation to the model with Gaussian random seeds with the same
three quantities. The construction does not depend on $I$, so applying
it at every arm count preserves correctness and every
permutation average.
\end{proof}

In particular, every $\delta$-correct algorithm specified by $K$
satisfies
\begin{equation}\label{bridge:kernel-lower}
 \cL(I,\delta)\le\frac1{n!}\sum_{\pi\in\mathfrak S_n}T_K(\pi I),
\end{equation}
so the lower bound of Theorem~\ref{thm:entropy} applies to the direct
costs of the probability-kernel formulation. This is the comparison direction used
in this paper; it does not identify the infimum over families of probability kernels
with the infimum over randomized algorithms, and no converse
representation of every randomized algorithm by history-dependent probability kernels is
claimed. The formal development (Appendix~\ref{verif:section})
establishes the finite-time state laws and the three conclusions above,
not a separate equality of measures on an independently constructed
infinite trace space.

\subsection{The logarithmic convention}
\label{bridge:convention}

The additive term of Conjecture~3.2 of \citet{chen2016open} is written
there as $\Delta_{[2]}^{-2}\ln\ln\Delta_{[2]}^{-1}$. This expression is
negative for $e^{-1}<\Delta_{[2]}<1$ and undefined at the admissible gap
$\Delta_{[2]}=1$, so a positive convention is needed. The paper uses
$\ell(D)=\log\bigl(e+\log(e+D)\bigr)$ with $D=\Delta_{[2]}^{-2}$,
see~\eqref{eq:notation-main}; Section~\ref{sec:model} also introduces
$g(\Delta)=\max\bigl\{1,\log\log(1/\Delta)\bigr\}$ for $0<\Delta<1$ and $g(1)=1$.
Both agree with the untruncated expression to first order as the gap
vanishes: with $x=\log(1/\Delta)$,
$\log(e+\Delta^{-2})=2x+\log(1+e\Delta^2)=2x+o(1)$, and one more
logarithm gives
\begin{equation}\label{bridge:expansion}
 \ell(\Delta^{-2})=\log\log(1/\Delta)+\log2+o(1)\qquad(\Delta\to0).
\end{equation}
Hence $\Delta^{-2}\ell(\Delta^{-2})$ and $\Delta^{-2}\log\log(1/\Delta)$
have the same asymptotic order. On the whole admissible range the two
positive conventions are equivalent up to absolute constants:
\begin{equation}\label{bridge:sandwich}
 g(\Delta)\le\ell(\Delta^{-2})\le7g(\Delta)\qquad(0<\Delta\le1).
\end{equation}
To prove~\eqref{bridge:sandwich}, let $x=\log(1/\Delta)\ge0$, so that
$D=\Delta^{-2}=e^{2x}\ge1$. Since $e<3D$, $\log(e+D)\le\log(4D)\le3+2x$,
and therefore $1<\ell(D)\le\log(6+2x)$. For the lower bound, if $x>0$
then $e+\log(e+D)>\log D=2x>x$, so $\ell(D)>\log x=\log\log(1/\Delta)$;
together with $\ell(D)>1$ this gives $\ell(D)\ge g(\Delta)$, and at
$\Delta=1$ we have $g(1)=1<\ell(1)$. For the upper bound, if $x\le1$
then $\ell(D)\le\log8<6\le7g(\Delta)$; if $x>1$ then $6+2x\le8x$, so
$\ell(D)\le\log8+\log x\le6+\log x\le7\max\{1,\log x\}=7g(\Delta)$.
Theorem~\ref{thm:universal} may therefore be stated equivalently with
$Dg(\Delta_{[2]})$ in place of $D\ell(D)$. No comparison of the form
$D\ell(D)\le D\log\log(1/\Delta_{[2]})+O(D)$ holds uniformly, because
the right-hand side tends to $-\infty$ as $\Delta_{[2]}\to1$.

\begin{corollary}[Small gaps and gaps bounded away from zero]\label{bridge:regimes}
Let $\Delta=\Delta_{[2]}$ and $0<\delta<0.1$. The algorithm of
Theorem~\ref{thm:universal} satisfies
\begin{align}
 T_A(I)&=O\bigl(\cL(I,\delta)+\Delta^{-2}\log\log(1/\Delta)\bigr)
 &&(0<\Delta\le e^{-e}),\label{bridge:small-gap}\\
 T_A(I)&=O\bigl(\cL(I,\delta)\bigr)
 &&(e^{-e}<\Delta\le1).\label{bridge:regular-gap}
\end{align}
\end{corollary}
\begin{proof}
Put $D=\Delta^{-2}$. By Theorem~\ref{thm:entropy},
$H\bigl[\log(1/\delta)+\Ent(I)\bigr]=O\bigl(\cL(I,\delta)\bigr)$, so
Theorem~\ref{thm:universal} gives $T_A(I)=O\bigl(\cL(I,\delta)+D\ell(D)\bigr)$.

In the first regime let $x=\log(1/\Delta)\ge e$, so $D=e^{2x}>e$. Then
$e+D<2D$ and $e+\log(e+D)<e+\log2+2x\le4x$, the last step because
$e+\log2<2e\le2x$. Hence $\ell(D)<\log(4x)=\log4+\log x\le3\log x$,
using $\log4<2\le2\log x$. Thus
$D\ell(D)\le3\Delta^{-2}\log\log(1/\Delta)$, which
gives~\eqref{bridge:small-gap}. At the endpoint $\Delta=e^{-e}$ the
double logarithm equals one.

In the second regime $1\le D<e^{2e}$. The function $\ell$ is increasing
and $e+e^{2e}<2e^{2e}$, so
$\ell(D)\le\ell(e^{2e})<\log(3e+\log2)<\log(4e)<3$. Hence
$D\ell(D)<3D\le3H$, because $D$ is one of the summands of $H$. Since
$\log(1/\delta)>\log10$ and $\Ent(I)\ge0$,
$H\le H\bigl[\log(1/\delta)+\Ent(I)\bigr]/\log10=O\bigl(\cL(I,\delta)\bigr)$ by
Theorem~\ref{thm:entropy}, which gives~\eqref{bridge:regular-gap}. The
two-arm term is thus absorbed into the instance-wise lower bound in the final
bound; it does not disappear from the intermediate cost analysis of
Appendix~\ref{cost:section}.
\end{proof}

Corollary~\ref{bridge:regimes} states the untruncated term on an
explicit small-gap range where it is at least one and asymptotically
active, and the bound in terms of the instance-wise lower bound alone on the remaining bounded range. For
$e^{-e}<\Delta\le e^{-1}$, \eqref{bridge:regular-gap} also
implies~\eqref{bridge:small-gap} because $\log\log(1/\Delta)\ge0$ there.
The corollary does not assert the literal untruncated display for every
gap: that expression is negative for $e^{-1}<\Delta<1$ and undefined at
$\Delta=1$. By Proposition~\ref{bridge:benchmark}, both displays hold
verbatim with $\cL_{\mathrm{SB}}(I,\delta)$ in place of $\cL(I,\delta)$.

Explicit conventions of this kind appear in the related literature.
\citet{jamieson2014lil} use in their sample complexity bound the truncated function $\log\log_+(x)$, equal to $\log\log x$ for
$x\ge e$ and to $0$ otherwise, so that it is defined and nonnegative
throughout $0<\Delta_i\le1$, and
\citet{chen2017towards} work in a regime in which the confidence
parameter and all gaps tend to zero. Section~\ref{sec:model} states the
present boundary convention explicitly; it is not attributed to the
problem statement of \citet{chen2016open}.

\section{Machine-checked verification}
\label{verif:section}

The definitions of Section~\ref{sec:model} and the main statements of
this paper have been formalized and proved in the Lean~4 proof
assistant \citep{demoura2021lean} on top of the Mathlib library,
whose original design is described by \citet{mathlib2020library}.
The development accompanies this
manuscript. This appendix states what is verified, how the
formal statements relate to the theorems as printed, and where the
formal proofs deviate from the presentation in
Appendices~\ref{sec:prelim}--\ref{cost:section}.

\paragraph{Formal model.}
An instance is a vector in $[0,1]^n$ with $n\ge2$ and a unique maximum.
An algorithm is a measurable function of an internal random seed, distributed as an
infinite sequence of independent standard Gaussian variables, and of the
finite observation history, with values in the sample-or-return decision
space; rewards form an independent Gaussian table and a sampling decision
appends the requested entry to the history. Success is the event that the
best label is returned at some finite time; the sample count is the sum
of the sampling indicators, with value $+\infty$ on nonreturning paths;
its expectation is the unconditional integral in $[0,\infty]$. $\delta$-correctness quantifies over every instance of every arm count,
and the instance-wise lower bound is the infimum over $\delta$-correct algorithm families of
the average over all $n!$ relabelings, exactly as in
\eqref{eq:benchmark}. The gap groups, $H$, $D$, and $\Ent(I)$ are
defined as in \eqref{eq:gap-entropy}--\eqref{eq:notation-main}. An independently written
transcription of these definitions and of the target statements was
checked, by definitional unfolding, to coincide with the definitions used
in the proofs.

\paragraph{Verified statements.}
For every $0<\delta<1/10$ and every instance:
\begin{enumerate}
\item the lower bound
 $\cL(I,\delta)\ge H(I)\bigl[\log(1/\delta)+\Ent(I)\bigr]/5$ of
 Theorem~\ref{lower:theorem} in its weaker $1/5$ form, and the upper
 bound $\cL(I,\delta)\le6.4\cdot10^{13}\,H(I)\bigl[\log(1/\delta)+\Ent(I)\bigr]$
 of Proposition~\ref{pointwise:theorem}, hence
 Theorem~\ref{thm:entropy};
\item one algorithm family, receiving only $\delta$ and $n$, that is
 $\delta$-correct, terminates almost surely, has finite expected sample complexity, and satisfies
 $T_A(I)\le12(10^{12}+1)\,\Bigl(H(I)\bigl[\log(1/\delta)+\Ent(I)\bigr]+D\ell(D)\Bigr)$,
 hence Theorem~\ref{thm:universal};
\item the consequence
 $T_A(I)\le C\bigl(\cL(I,\delta)+D\ell(D)\bigr)$ for a universal $C$,
 both with $D\ell(D)$ and with $D\,g(\Delta_{[2]})$ for the truncation
 $g$ of Section~\ref{sec:model}, the two versions being proved
 equivalent up to the factor $7$, together with the two-regime statement of
 Corollary~\ref{bridge:regimes};
\item the results of Appendix~\ref{bridge:section}: equality of the
 instance-wise lower bound over algorithms with arbitrary standard Borel internal randomness, the implementation of history-dependent probability kernels with
 identical success probability, cost, and termination, the completion of
 algorithms with an explicit abandonment action, and the deterministic
 nature of the two upper-bound algorithms.
\end{enumerate}
The literal untruncated expression
$\Delta_{[2]}^{-2}\ln\ln\Delta_{[2]}^{-1}$ is stated separately in the
formal development and is not proved; as explained in
Section~\ref{sec:model}, it cannot hold uniformly near
$\Delta_{[2]}=1$.

\paragraph{Relation to the printed proofs.}
The appendices contain complete proofs that do not depend on the
formal development. The formal proofs establish the statements above,
not the prose of the appendices sentence by sentence, and three
deviations are deliberate.
First, the formal lower bound does not construct a symmetrized algorithm
as in Appendix~\ref{lower:section}; it works directly with
permutation-averaged numbers of samples of individual labels and with the
bijection $\pi\mapsto(a\,d)\circ\pi$ on labelings in the change of distribution that modifies two arm means. Second, Lemma~\ref{sound:domination} is proved in
Appendix~\ref{sound:section} by a pathwise coupling to independent
Poisson clocks; the formal development instead proves the weaker but
sufficient statement that, for every increasing family of subsets of
arms, the set of arms whose empirical means satisfy \eqref{alg:severe-flag} in at least one call is stochastically
dominated by independent Bernoulli$(p)$ indicators, through a
bound on the cumulative predictable hazard and a supermartingale argument. All uses of
the lemma in Appendix~\ref{sound:section} involve only increasing
events. Third, the formally verified algorithms are variants of the algorithms of
Appendices~\ref{pointwise:section} and~\ref{alg:section} with deterministic
sample counts per run:
every run of the single algorithm requests exactly its sample budget
$Q_j$, padding with requests to a fixed arm after a return or an abort,
and the sequence of runs reports an answer at the following block boundary;
the algorithm with a fixed target instance pads each run to a deterministic budget
computed from the hard-coded $H(I)$ and $\Ent(I)$. These variants only
change absolute constants and are covered by the same arguments. The
single algorithm is instantiated with $C_Q=10^{12}$, and the bound on the sum of exponentials \eqref{lower:kraft} is proved with the slightly better constant
$(79/9)\delta$.

\paragraph{Checks performed.}
The development (about $195$ modules and $4800$ declarations for the
main results, and $450$ declarations for Appendix~\ref{bridge:section})
was compiled from source against a pinned Mathlib revision in a fresh
directory. A recursive audit of every declaration found only the axioms
\texttt{propext}, \texttt{Classical.choice}, and \texttt{Quot.sound}, and
no use of \texttt{sorry}. The complete import closure was replayed in a
fresh kernel environment with the reference checker shipped with the
Lean toolchain. The independently transcribed statements were compared
with the proved statements using the Lean FRO \texttt{comparator} tool,
which also checked the permitted axioms and replayed the proofs through
the Lean kernel and through the independent external kernel
\texttt{nanoda}. None of these checks verifies the informal text, and
the development has not been reviewed by a person outside the project.

\bibliographystyle{plainnat}
\bibliography{ref}

@article{aronow2026positive,
  author = {Aronow, P. M. and Kallus, Nathan and Lopatto, Patrick},
  title = {{A positive resolution of the gap-entropy conjecture}},
  journal = {arXiv preprint arXiv:2609.10529v1},
  year = {2026},
  url = {https://arxiv.org/abs/2609.10529v1}
}

@inproceedings{chen2016open,
  author = {Chen, Lijie and Li, Jian},
  title = {{Open Problem: Best Arm Identification: Almost Instance-Wise Optimality and the Gap Entropy Conjecture}},
  booktitle = {29th Annual Conference on Learning Theory},
  series = {Proceedings of Machine Learning Research},
  volume = {49},
  pages = {1643--1646},
  year = {2016},
  publisher = {PMLR},
  url = {https://proceedings.mlr.press/v49/chen16b.html}
}

@inproceedings{chen2017towards,
  author = {Chen, Lijie and Li, Jian and Qiao, Mingda},
  title = {{Towards Instance Optimal Bounds for Best Arm Identification}},
  booktitle = {Proceedings of the 2017 Conference on Learning Theory},
  series = {Proceedings of Machine Learning Research},
  volume = {65},
  pages = {535--592},
  year = {2017},
  publisher = {PMLR},
  url = {https://proceedings.mlr.press/v65/chen17b.html}
}

@inproceedings{karnin2013almost,
  author = {Karnin, Zohar and Koren, Tomer and Somekh, Oren},
  title = {{Almost Optimal Exploration in Multi-Armed Bandits}},
  booktitle = {Proceedings of the 30th International Conference on Machine Learning},
  series = {Proceedings of Machine Learning Research},
  volume = {28(3)},
  pages = {1238--1246},
  year = {2013},
  publisher = {PMLR},
  url = {https://proceedings.mlr.press/v28/karnin13.html}
}

@inproceedings{jamieson2014lil,
  author = {Jamieson, Kevin and Malloy, Matthew and Nowak, Robert and Bubeck, S{\'e}bastien},
  title = {{lil' UCB : An Optimal Exploration Algorithm for Multi-Armed Bandits}},
  booktitle = {Proceedings of The 27th Conference on Learning Theory},
  series = {Proceedings of Machine Learning Research},
  volume = {35},
  pages = {423--439},
  year = {2014},
  publisher = {PMLR},
  url = {https://proceedings.mlr.press/v35/jamieson14.html}
}

@inproceedings{garivier2016optimal,
  author = {Garivier, Aur{\'e}lien and Kaufmann, Emilie},
  title = {{Optimal Best Arm Identification with Fixed Confidence}},
  booktitle = {29th Annual Conference on Learning Theory},
  series = {Proceedings of Machine Learning Research},
  volume = {49},
  pages = {998--1027},
  year = {2016},
  publisher = {PMLR},
  url = {https://proceedings.mlr.press/v49/garivier16a.html}
}

@inproceedings{simchowitz2017simulator,
  author = {Simchowitz, Max and Jamieson, Kevin and Recht, Benjamin},
  title = {{The Simulator: Understanding Adaptive Sampling in the Moderate-Confidence Regime}},
  booktitle = {Proceedings of the 2017 Conference on Learning Theory},
  series = {Proceedings of Machine Learning Research},
  volume = {65},
  pages = {1794--1834},
  year = {2017},
  publisher = {PMLR},
  url = {https://proceedings.mlr.press/v65/simchowitz17a.html}
}

@inproceedings{kaufmann2013information,
  author = {Kaufmann, Emilie and Kalyanakrishnan, Shivaram},
  title = {{Information Complexity in Bandit Subset Selection}},
  booktitle = {Proceedings of the 26th Annual Conference on Learning Theory},
  series = {Proceedings of Machine Learning Research},
  volume = {30},
  pages = {228--251},
  year = {2013},
  publisher = {PMLR},
  url = {https://proceedings.mlr.press/v30/Kaufmann13.html}
}

@inproceedings{kaufmann2014testing,
  author = {Kaufmann, Emilie and Capp{\'e}, Olivier and Garivier, Aur{\'e}lien},
  title = {{On the Complexity of A/B Testing}},
  booktitle = {Proceedings of The 27th Conference on Learning Theory},
  series = {Proceedings of Machine Learning Research},
  volume = {35},
  pages = {461--481},
  year = {2014},
  publisher = {PMLR},
  url = {https://proceedings.mlr.press/v35/kaufmann14.html}
}

@inproceedings{carpentier2016tight,
  author = {Carpentier, Alexandra and Locatelli, Andrea},
  title = {{Tight (Lower) Bounds for the Fixed Budget Best Arm Identification Bandit Problem}},
  booktitle = {29th Annual Conference on Learning Theory},
  series = {Proceedings of Machine Learning Research},
  volume = {49},
  pages = {590--604},
  year = {2016},
  publisher = {PMLR},
  url = {https://proceedings.mlr.press/v49/carpentier16.html}
}

@inproceedings{shang2020fixed,
  author = {Shang, Xuedong and de Heide, Rianne and Menard, Pierre and Kaufmann, Emilie and Valko, Michal},
  title = {{Fixed-confidence guarantees for Bayesian best-arm identification}},
  booktitle = {Proceedings of the Twenty Third International Conference on Artificial Intelligence and Statistics},
  series = {Proceedings of Machine Learning Research},
  volume = {108},
  pages = {1823--1832},
  year = {2020},
  publisher = {PMLR},
  url = {https://proceedings.mlr.press/v108/shang20a.html}
}

@inproceedings{you2023information,
  author = {You, Wei and Qin, Chao and Wang, Zihao and Yang, Shuoguang},
  title = {{Information-Directed Selection for Top-Two Algorithms}},
  booktitle = {Proceedings of Thirty Sixth Conference on Learning Theory},
  series = {Proceedings of Machine Learning Research},
  volume = {195},
  pages = {2850--2851},
  year = {2023},
  publisher = {PMLR},
  url = {https://proceedings.mlr.press/v195/you23a.html}
}

@inproceedings{huang2017structured,
  author = {Huang, Ruitong and Ajallooeian, Mohammad M. and Szepesv{\'a}ri, Csaba and M{\"u}ller, Martin},
  title = {{Structured Best Arm Identification with Fixed Confidence}},
  booktitle = {Proceedings of the 28th International Conference on Algorithmic Learning Theory},
  series = {Proceedings of Machine Learning Research},
  volume = {76},
  pages = {593--616},
  year = {2017},
  publisher = {PMLR},
  url = {https://proceedings.mlr.press/v76/huang17a.html}
}

@inproceedings{tao2018linear,
  author = {Tao, Chao and Blanco, Sa{\'u}l and Zhou, Yuan},
  title = {{Best Arm Identification in Linear Bandits with Linear Dimension Dependency}},
  booktitle = {Proceedings of the 35th International Conference on Machine Learning},
  series = {Proceedings of Machine Learning Research},
  volume = {80},
  pages = {4877--4886},
  year = {2018},
  publisher = {PMLR},
  url = {https://proceedings.mlr.press/v80/tao18a.html}
}

@inproceedings{russo2016simple,
  author = {Russo, Daniel},
  title = {{Simple Bayesian Algorithms for Best Arm Identification}},
  booktitle = {29th Annual Conference on Learning Theory},
  series = {Proceedings of Machine Learning Research},
  volume = {49},
  pages = {1417--1418},
  year = {2016},
  publisher = {PMLR},
  url = {https://proceedings.mlr.press/v49/russo16.html}
}

@inproceedings{gabillon2012unified,
  author = {Gabillon, Victor and Ghavamzadeh, Mohammad and Lazaric, Alessandro},
  title = {{Best Arm Identification: A Unified Approach to Fixed Budget and Fixed Confidence}},
  booktitle = {Advances in Neural Information Processing Systems},
  volume = {25},
  year = {2012},
  publisher = {Curran Associates, Inc.},
  url = {https://proceedings.neurips.cc/paper_files/paper/2012/hash/8b0d268963dd0cfb808aac48a549829f-Abstract.html}
}

@inproceedings{degenne2019games,
  author = {Degenne, R\'{e}my and Koolen, Wouter M and M\'{e}nard, Pierre},
  title = {{Non-Asymptotic Pure Exploration by Solving Games}},
  booktitle = {Advances in Neural Information Processing Systems},
  volume = {32},
  year = {2019},
  publisher = {Curran Associates, Inc.},
  url = {https://proceedings.neurips.cc/paper_files/paper/2019/hash/8d1de7457fa769ece8d93a13a59c8552-Abstract.html}
}

@inproceedings{jourdan2022top,
  author = {Jourdan, Marc and Degenne, R\'{e}my and Baudry, Dorian and de Heide, Rianne and Kaufmann, Emilie},
  title = {{Top Two Algorithms Revisited}},
  booktitle = {Advances in Neural Information Processing Systems},
  volume = {35},
  pages = {26791--26803},
  year = {2022},
  publisher = {Curran Associates, Inc.},
  doi = {10.52202/068431-1943},
  url = {https://proceedings.neurips.cc/paper_files/paper/2022/hash/ab5f5f22e3e09f4424592ffb06840ab0-Abstract-Conference.html}
}

@inproceedings{jourdan2023nonasymptotic,
  author = {Jourdan, Marc and Degenne, R\'{e}my},
  title = {{Non-Asymptotic Analysis of a UCB-based Top Two Algorithm}},
  booktitle = {Advances in Neural Information Processing Systems},
  volume = {36},
  pages = {68980--69020},
  year = {2023},
  publisher = {Curran Associates, Inc.},
  doi = {10.52202/075280-3019},
  url = {https://proceedings.neurips.cc/paper_files/paper/2023/hash/d9b564716709357b4bccec9fc9ad04d2-Abstract-Conference.html}
}

@inproceedings{chen2014combinatorial,
  author = {Chen, Shouyuan and Lin, Tian and King, Irwin and Lyu, Michael R. and Chen, Wei},
  title = {{Combinatorial Pure Exploration of Multi-Armed Bandits}},
  booktitle = {Advances in Neural Information Processing Systems},
  volume = {27},
  year = {2014},
  publisher = {Curran Associates, Inc.},
  url = {https://proceedings.neurips.cc/paper_files/paper/2014/hash/d3ea0f3316d2da934d79b8b344eafee4-Abstract.html}
}

@inproceedings{soare2014linear,
  author = {Soare, Marta and Lazaric, Alessandro and Munos, R\'{e}mi},
  title = {{Best-Arm Identification in Linear Bandits}},
  booktitle = {Advances in Neural Information Processing Systems},
  volume = {27},
  year = {2014},
  publisher = {Curran Associates, Inc.},
  url = {https://proceedings.neurips.cc/paper_files/paper/2014/hash/f8d84caae48546d0934d637ab54f7086-Abstract.html}
}

@inproceedings{jedra2020optimal,
  author = {Jedra, Yassir and Proutiere, Alexandre},
  title = {{Optimal Best-arm Identification in Linear Bandits}},
  booktitle = {Advances in Neural Information Processing Systems},
  volume = {33},
  pages = {10007--10017},
  year = {2020},
  publisher = {Curran Associates, Inc.},
  url = {https://proceedings.neurips.cc/paper_files/paper/2020/hash/7212a6567c8a6c513f33b858d868ff80-Abstract.html}
}

@inproceedings{degenne2019multiple,
  author = {Degenne, R\'{e}my and Koolen, Wouter},
  title = {{Pure Exploration with Multiple Correct Answers}},
  booktitle = {Advances in Neural Information Processing Systems},
  volume = {32},
  year = {2019},
  publisher = {Curran Associates, Inc.},
  url = {https://proceedings.neurips.cc/paper_files/paper/2019/hash/60cb558c40e4f18479664069d9642d5a-Abstract.html}
}

@article{even2006action,
  author = {Eyal Even-Dar and Shie Mannor and Yishay Mansour},
  title = {{Action Elimination and Stopping Conditions for the Multi-Armed Bandit and Reinforcement Learning Problems}},
  journal = {Journal of Machine Learning Research},
  volume = {7},
  number = {39},
  pages = {1079--1105},
  year = {2006},
  url = {http://jmlr.org/papers/v7/evendar06a.html}
}

@article{kaufmann2016complexity,
  author = {Emilie Kaufmann and Olivier Capp{{\'e}} and Aur{{\'e}}lien Garivier},
  title = {{On the Complexity of Best-Arm Identification in Multi-Armed Bandit Models}},
  journal = {Journal of Machine Learning Research},
  volume = {17},
  number = {1},
  pages = {1--42},
  year = {2016},
  url = {http://jmlr.org/papers/v17/kaufman16a.html}
}

@article{kaufmann2021mixture,
  author = {Emilie Kaufmann and Wouter M. Koolen},
  title = {{Mixture Martingales Revisited with Applications to Sequential Tests and Confidence Intervals}},
  journal = {Journal of Machine Learning Research},
  volume = {22},
  number = {246},
  pages = {1--44},
  year = {2021},
  url = {http://jmlr.org/papers/v22/18-798.html}
}

@inproceedings{demoura2021lean,
  author = {de Moura, Leonardo and Ullrich, Sebastian},
  title = {{The Lean 4 Theorem Prover and Programming Language}},
  booktitle = {Automated Deduction -- CADE 28: 28th International Conference on Automated Deduction, Virtual Event, July 12--15, 2021, Proceedings},
  pages = {625--635},
  year = {2021},
  publisher = {Springer-Verlag},
  doi = {10.1007/978-3-030-79876-5_37},
  url = {https://doi.org/10.1007/978-3-030-79876-5_37}
}

@article{mannor2004sample,
  author = {Shie Mannor and John N. Tsitsiklis},
  title = {{The Sample Complexity of Exploration in the Multi-Armed Bandit Problem}},
  journal = {Journal of Machine Learning Research},
  volume = {5},
  pages = {623--648},
  year = {2004},
  url = {https://www.jmlr.org/papers/v5/mannor04b.html}
}

@article{farrell1964,
  author = {R. H. Farrell},
  title = {{Asymptotic Behavior of Expected Sample Size in Certain One Sided Tests}},
  journal = {The Annals of Mathematical Statistics},
  volume = {35},
  number = {1},
  pages = {36--72},
  year = {1964},
  doi = {10.1214/aoms/1177703731},
  url = {https://doi.org/10.1214/aoms/1177703731}
}

@inproceedings{even2002pac,
  author = {Eyal Even-Dar and Shie Mannor and Yishay Mansour},
  title = {{{PAC} Bounds for Multi-armed Bandit and {Markov} Decision Processes}},
  booktitle = {Computational Learning Theory},
  pages = {255--270},
  year = {2002},
  publisher = {Springer},
  doi = {10.1007/3-540-45435-7_18},
  url = {https://doi.org/10.1007/3-540-45435-7_18}
}

@inproceedings{mathlib2020library,
  author = {{The mathlib community}},
  title = {{The {L}ean mathematical library}},
  booktitle = {Proceedings of the 9th {ACM} {SIGPLAN} International Conference on Certified Programs and Proofs, {CPP} 2020, New Orleans, LA, USA, January 20-21, 2020},
  pages = {367--381},
  year = {2020},
  doi = {10.1145/3372885.3373824},
  url = {https://doi.org/10.1145/3372885.3373824}
}

@article{howard2021sequences,
  author = {Steven R. Howard and Aaditya Ramdas and Jon McAuliffe and Jasjeet Sekhon},
  title = {{Time-uniform, nonparametric, nonasymptotic confidence sequences}},
  journal = {The Annals of Statistics},
  volume = {49},
  number = {2},
  pages = {1055--1080},
  year = {2021},
  doi = {10.1214/20-aos1991},
  url = {https://doi.org/10.1214/20-aos1991}
}

@inproceedings{audibert2010best,
  author = {Audibert, Jean-Yves and Bubeck, S{\'e}bastien and Munos, R{\'e}mi},
  title = {{Best Arm Identification in Multi-Armed Bandits}},
  booktitle = {Proceedings of the Twenty-third Conference on Learning Theory (COLT 2010)},
  pages = {41--53},
  year = {2010},
  publisher = {Omnipress},
  url = {https://www.learningtheory.org/colt2010/papers/59Audibert.pdf}
}

@article{chen2015optimal,
  author = {Chen, Lijie and Li, Jian},
  title = {{On the Optimal Sample Complexity for Best Arm Identification}},
  journal = {arXiv preprint arXiv:1511.03774},
  year = {2015},
  url = {https://arxiv.org/abs/1511.03774}
}

\end{document}